\documentclass[twoside,11pt]{article}

\usepackage{blindtext}

\usepackage[abbrvbib, preprint]{jmlr2e}

\usepackage{amsthm}
\usepackage{physics}
\usepackage{algorithm,algorithmicx,algpseudocode}
\usepackage{bbm}

\usepackage[nameinlink, noabbrev]{cleveref}
\newtheorem{theorem}{Theorem}%  meant for continuous numbers
\newtheorem{proposition}{Proposition}%
\newtheorem{lemma}{Lemma}
\newtheorem{remark}{Remark}%
\newtheorem{example}{Example}%

\newtheorem*{scc}{Strong Central Condition}
\newtheorem*{cc}{Unit-Dominance Condition}

\Crefname{corollary}{Corollary}{Corollaries}
\Crefname{lemma}{Lemma}{Lemmas}
\Crefname{figure}{Figure}{Figures}
\Crefname{definition}{Definition}{Definitions}
\Crefname{inequality}{inequality}{inequalities}
\Crefname{example}{Example}{Examples}
\Crefname{proposition}{Proposition}{Propositions}
\Crefname{theorem}{Theorem}{Theorems}

\usepackage{tikz}
\usetikzlibrary{decorations.pathmorphing}
\usepackage{subcaption}

\newcommand{\N}{\mathbb{N}}
\newcommand{\Z}{\mathbb{Z}}

\newcommand{\R}{\mathbb{R}}

\renewcommand{\implies}{\Longrightarrow}

\newcommand{\altgiven}{;\,}
\newcommand{\iid}{\overset{iid}{\sim}}
\DeclareMathOperator*{\argmin}{\arg\,\min}
\DeclareMathOperator*{\argmax}{\arg\,\max}

\DeclareMathOperator{\E}{E}
\DeclareMathOperator{\prob}{P}
\DeclareMathOperator*{\Var}{Var}
\renewcommand{\var}{\Var}

\renewcommand{\ip}{\overset{p}{\longrightarrow}}

\newcommand{\D}{\mathcal{D}}

\newcommand{\X}{\mathcal{X}}

\newcommand{\rhat}{\widehat{R}}
\newcommand{\thetahat}{\widehat{\theta}}
\newcommand{\given}{\mid}
\newcommand{\bigmid}{\, \Bigl\vert \,}
\newcommand{\omegahat}{\widehat{\omega}}
\newcommand{\Ghat}{\widehat{G}}

\newcommand{\omegaunder}{\underline{\Omega}}

\usepackage{lastpage}
\ShortHeadings{Generalized Universal Inference on Risk Minimizers}{Dey, Martin, and Williams}
\firstpageno{1}

\begin{document}

\title{Generalized Universal Inference on Risk Minimizers}

\author{\name Neil Dey \email ndey3@ncsu.edu\\
\name Ryan Martin \email rgmarti3@ncsu.edu \\
\name Jonathan P. Williams \email jwilli27@ncsu.edu \\
       \addr Department of Statistics\\
       North Carolina State University\\
       Raleigh, NC 27607-6698, USA
}

\maketitle

\begin{abstract}%   <- trailing '%' for backward compatibility of .sty file
A common goal in statistics and machine learning is estimation of unknowns. Point estimates alone are of little value without an accompanying measure of uncertainty, but traditional uncertainty quantification methods, such as confidence sets and p-values, often require distributional or structural assumptions that may not be justified in modern applications. The present paper considers a very common case in machine learning, where the quantity of interest is the minimizer of a given risk (expected loss) function. We propose a generalization of universal inference specifically designed for inference on risk minimizers.  Notably, our generalized universal inference attains finite-sample frequentist validity guarantees under a condition common in the statistical learning literature. 
One version of our procedure is also anytime-valid, i.e., 
it maintains the finite-sample validity properties regardless of the stopping rule used for the data collection process.
Practical use of our proposal requires tuning, and we offer a data-driven procedure with strong empirical performance across a broad range of challenging statistical and machine learning examples. 
\end{abstract}


\section{Introduction}
\label{sec:intro}

In statistics and machine learning applications, many sophisticated and computationally efficient procedures have been developed for estimating high- or even infinite-dimensional unknowns with strong theoretical support in the form of asymptotic consistency.  It is important, however, to accompany these estimates with an appropriate measure of uncertainty, typically in the form of a confidence set or in the form of p-values associated with relevant hypotheses tests. Ensuring the reliability of uncertainty quantification is a major challenge, largely because the available theory generally requires distributional assumptions or structural simplifications that the data scientist is reluctant or unable to make. As such, there is a need for general strategies that provide provably valid uncertainty quantification in modern, high-dimensional problems for finite sample sizes.

Our paper considers a broad class of statistical learning problems where the quantity of interest is defined as a {\em risk minimizer}, i.e., the minimizer of a risk (expected loss) function.  This includes the typical regression and classification problems common in machine learning, with squared-error and zero-one loss functions, respectively.  It also covers typical cases in the statistics literature where the quantity of interest is the parameter of a correctly- or incorrectly-specified statistical model, with loss corresponding to the model's negative log-density/mass function.  The last two classes of problems fall squarely in the machine learning and statistics domains, respectively, but there are others whose classification is less clear---such as quantile regression \citep{koenker.bassett.1978}, inference on the minimum clinically important difference \citep{xu.mcid, syring.martin.mcid}, etc.---and these too are covered by our proposed framework.  Much of the extant statistical literature on uncertainty quantification for risk minimizers comes from the broad area of robust statistics \citep[e.g.,][]{huber1981,hampel2011}, in particular, the well-studied area of M-estimation \citep[e.g.,][]{huber1981,maronna2006,boos2018}. The classical results in this area impose regularity conditions and achieve only asymptotically valid frequentist inference. More recent results \citep[e.g.,][]{hudson2021, cella2022} require fewer regularity conditions but still only offer asymptotic validity.

An important recent development was the so-called {\em universal inference} framework of \citet{wasserman2020}. They present simple and elegant procedures that offer provably valid uncertainty quantification (e.g., confidence sets and p-values) under virtually no conditions and without the need for asymptotic approximations. One of their main results is based on a clever use of data-splitting to construct a ``split likelihood ratio" for which finite-sample distributional bounds on error rates can be obtained under no regularity conditions. Their focus was limited, however, to settings in which a likelihood function is available---that is, to problems characterized by a correctly specified statistical model. This limitation is partially addressed by \citet{park2023} who allow for misspecification of the statistical model and construct valid inference on the parameter value that minimizes the Kullback--Leibler divergence of the posited statistical model from the true data-generating distribution. However, as mentioned above, assuming a statistical model is a non-trivial restriction in many practical applications, so there is a practical need for methods with provable finite-sample validity beyond model-based settings. Our main contribution here is an extension of the developments in \citet{wasserman2020} that covers many learning problems beyond those determined by a statistical model. 

Our proposed {\em generalized universal inference} framework replaces the log-likelihood function in the model-based universal inference framework in \citet{wasserman2020} with the empirical risk function, an essential ingredient in the learning problem.  Wasserman et al.'s developments leaned on the simple, well-known fact that likelihood ratios have expected value 1.  In our present context, however, there is no likelihood ratio and, therefore, no direct analogue of Wasserman et al.'s key property to apply.  We overcome this obstacle by identifying a single regularity condition---namely, the ``strong central condition,'' common in the statistical learning literature \citep[e.g.,][]{vanErvan2015,grunwald2020,syring2023}---sufficient for showing that Wasserman et al.'s relevant property also holds for our proposed generalized universal inference over a broad class of learning problems.  Beyond (anytime-)validity, we also establish results concerning the statistical efficiency of our proposal, e.g., asymptotic power of our test procedures.  

The remainder of this paper is organized as follows.  After some detailed background and problem setup in \Cref{sec:background}, we present our generalized universal inference framework 
%that links statistical learning to safe inference 
in \Cref{sec:statistic}, and we state its theoretical validity properties.  An important, practical, and novel detail in our approach is the choice of a suitable {\em learning rate} parameter. Choosing the learning rate is a challenging problem and, in \Cref{sec:learning_rate}, we provide theoretical results on data-driven approaches to learning rate selection, and propose a selection strategy that empirically maintains (anytime-)validity. We then further give theory on the efficiency of our framework given an appropriate choice of learning rate. In \Cref{sec:simulation}, we present simulation studies that demonstrate the dual validity and efficiency of our proposed approach in a variety of challenging settings. In particular, we compare our method to that of \citet{waudbysmith2023}, which is designed specifically for (anytime-valid) nonparametric inference on the mean of an unknown distribution, and empirically demonstrate our proposal's superior efficiency. Furthermore, we show how it addresses various factors contributing to the replication crisis in science. In \Cref{sec:realdata}, we demonstrate how our approach performs in real data examples, namely, Millikan's classic experiment to measure the charge of an electron, as well as quantile estimation of user ratings from the website MyAnimeList.  We finish with concluding remarks in \Cref{sec:conclusion}. Proofs of all the theorems can be found in \Cref{sec:appendix}. The code for reproducing the simulation experiments presented in this paper is available at \url{https://github.com/neil-dey/universal-inference}.

\section{Problem setup and related work}
\label{sec:background}

Suppose that the observable data $Z^n := (Z_1, \ldots, Z_n)$ are i.i.d.~from an unknown distribution $\D$ over a set $\mathbb{Z}$. A loss function $\ell:\Theta\times\mathbb{Z} \to \R^+$ is chosen by a practitioner that measures how well a parameter $\theta\in\Theta$ conforms with an observed datum $z \in \mathbb{Z}$; small $\ell(\theta,z)$ values indicate greater conformity between $z$ and $\theta$. We write $R(\theta) := \E_{Z \sim \D}\{ \ell(\theta; Z) \}$ for the {\em risk} or expected loss function.  Our goal is to infer the risk minimizer 
\begin{equation*}
    \theta^* := \argmin_{\theta\in\Theta} R(\theta).
\end{equation*}
Of course, $\theta^*$ is unknown because the distribution $\D$ is unknown, 
%This is impossible without knowledge of the distribution of $\D$, 
but we can estimate $\theta^*$ using the data $Z^n$ from $\D$. 
 To this end, the \textit{empirical risk minimizer} (ERM) is 
\begin{equation*}
    \thetahat_n := \argmin_{\theta\in\Theta} \rhat_n(\theta), 
    %\equiv \argmin_{\theta \in \Theta} \frac{1}{n}\sum_{i=1}^n \ell(\theta\altgiven z_i),
\end{equation*}
where $\rhat_n(\theta) := n^{-1} \sum_{i=1}^n \ell(\theta; Z_i)$ is the {\em empirical risk} function.  The intuition is that $\rhat_n$ should be close to $R$, at least for large $n$, so the ERM $\thetahat_n$ should be close to $\theta^*$. It may happen that the ERM does not exist or it is difficult or inefficient to compute exactly; in such cases, it may be useful to instead compute an \textit{almost-ERM} (AERM)---that is, an estimator $\thetahat_n$ satisfying
\begin{equation*}
    \rhat_n(\thetahat_n) \leq \inf_{\theta\in\Theta} \rhat_n(\theta) + \frac{\delta}{n^{1+\varepsilon}}
\end{equation*}
for some fixed nonnegative constants $\varepsilon$ and $\delta$. When the constants $\varepsilon$ and $\delta$ are of import, we specify that $\thetahat_n$ is a $(\varepsilon, \delta)$-AERM. AERMs are fairly easy to construct: for example, a $(0, \delta)$-AERM under the $L^2$ loss only requires convergence of the estimator to the ERM at rate $n^{-1/2}$ on the original $\Theta$ space---a rate that can often be obtained for approximation schemes such as stochastic gradient descent \citep{nemirovski2009}.

A variety of approaches are available to quantify the uncertainty about $\theta^\star$ in the ERM $\thetahat_n$.  As mentioned in \Cref{sec:intro}, classical solutions offer asymptotic frequentist guarantees under rather strong regularity conditions.  It is demonstrated in \citet{wasserman2020}, on the other hand, that with a well-specified model $\{\prob_{\theta}\mid\theta\in\Theta\}$ featuring a likelihood function $L(\theta; Z^n)$, one can obtain confidence sets for $\theta^*$ with no regularity conditions.  One of their proposed strategies is {\em sample splitting}.  That is, partition the sample $Z^n$ into sub-samples $Z^{(1)}$ and $Z^{(2)}$ and compute the maximum likelihood estimator $\thetahat^{(1)}$ on $Z^{(1)}$; then a $1-\alpha$ level confidence set for $\theta^*$ is given by $\{ \theta\in\Theta \mid T(\theta) \leq \alpha^{-1} \}$, where
\begin{equation*}
T(\theta) = T(\theta; Z^{(1)}, Z^{(2)}) := \frac{L(\thetahat^{(1)}; Z^{(2)})}{L(\theta; Z^{(2)})}
\end{equation*}
is called the ``split likelihood-ratio" for obvious reasons. It is notable that $T(\theta)$ is an example of an \textit{e-value}, defined by the property that $\E_{\theta^*}\{T(\theta^*)\} \leq 1$, where $\E_{\theta^*}$ denotes the expected value under the assumption that the data $Z^n$ was generated from $\prob_{\theta^*}$.

The notion of an e-value can be traced as far back as \citet{wald1945, wald.sequential}, but there has been a surge of interest recently \citep[e.g.,][]{vovk2021, howard2021, xu2021, ramdas2023} for at least two reasons. First, the reciprocal of an e-value is a p-value (i.e., is stochastically no smaller than a uniform random variable) by Markov's inequality, so e-values can readily be used for uncertainty quantification. Second, e-values have several benefits as ``measures of evidence" over general p-values.  For example, while it is not clear how to combine p-values from independent tests, it is clear that taking the product of independent e-values is itself an e-value. Furthermore, this product of independent e-values remains an e-value under optional continuation---the practice of deciding whether or not to continue collecting new data and conducting further independent tests based on the outcomes of previous tests---and so has practical use in meta-analyses \citep{grunwald2020b}.
Additionally, e-values also tend to be more robust to model misspecification and dependence compared to general p-values; see \citet{wang2022} and \citet{ramdas2025}.
%e.g., it is shown in \citet{wang2022} that applying the Benjamini--Hochberg procedure to e-values maintains control over false discovery rates even for arbitrary dependence between the e-values, whereas the same cannot be said for p-values. 
However, e-values are not a direct upgrade to p-values: their safety guarantees imply that uncertainty quantification with e-values tends to be more conservative than that with p-values.

A closely related notion is that of an \textit{e-process}, i.e., a non-negative supermartingale $(E_n)_{n\in\N}$ such that $\E(E_\tau) \leq 1$ under the null hypothesis for any stopping time $\tau$ \citep{shafer2011,ramdas2023,ruf2023}. It is clear that any stopped e-process is also an e-value and thus inherits the relevant benefits. Additionally, the definition of an e-process yields an ``anytime-validity" property: If $(E_n)_{n\in\N}$ is an e-process, the reciprocal of $\max_{n=1,\ldots, \tau} E_n$ remains a p-value for any stopping time $\tau$ \citep{ramdas2023}. That is, the sample size need not be fixed ahead of time, and one may even choose whether or not to collect more data based on what has been observed ``up to that point." This is in stark contrast to a standard p-value, which generally depends on fixing a sample size ahead of time and prohibits any sort of data-snooping; but see \citet[][Sec.~6]{martin.basu}. Because peeking at the data to decide whether to stop or continue data-collection is common in science, the use of anytime-valid measures of evidence such as e-processes is highly desirable.

How does one construct an e-process? Like the e-value described above, these take the general form of likelihood ratios but with a sequential flavor \citep[e.g.,][Eq.~10.10]{wald.sequential}.  A general proposal was given in \citet[][Sec.~8]{wasserman2020} and particular instantiations have been put forward in, e.g., \citet{gangrade2023sequential} and \citet{dixit2023}; see, also, the survey in \citet{ramdas2023}. In particular, as an alternative to sample splitting described above, consider lagged estimators 
\[ \thetahat_k = \argmax_{\theta \in \Theta} L(\theta; Z^k), \quad k=1,2,\ldots\]
and the corresponding ``running likelihood-ratio" test statistic
\begin{equation}
\label{eq:running.lr}
M_n(\theta) := \frac{\prod_{i=1}^n L(\thetahat_{i-1}; Z^i)}{\prod_{i=1}^n L(\theta; Z^i)},
\end{equation}
where $\thetahat_0$ is a fixed constant. 
Then $\{M_n(\theta^*)\}_{n\in\N}$ is an e-process and, therefore, provides anytime-valid inference on $\theta^*$. 

\section{Generalized universal inference} \label{sec:statistic}

\subsection{GUe-value construction}
If the data-generating distribution $\D$ is unknown, or if the quantity of interest is not defined as the parameter that determines a statistical model (and is rather defined as the minimizer of a more general risk function), then the approach of \citet{wasserman2020} is not directly applicable. To deal with the general statistical learning problem, we propose the following generalized universal inference framework. To start, define the online {\em generalized universal e-value} (GUe-value, pronounced ``gooey-value"):
\begin{equation}
\label{eq:online}
    G_{n, \text{on}}(\theta)  := \exp\Bigl[-\omega  \sum_{i=1}^n \{ \ell(\thetahat_{i-1}; Z_i) - \ell(\theta; Z_i) \} \Bigr], 
\end{equation}
where $\thetahat_k$ is any AERM on the first $k$ sample elements, with $\thetahat_0$ a pre-specified constant, and $\omega\geq 0$ is a \textit{learning rate} discussed in detail in \Cref{sec:learning_rate}. The right-hand side of the above display is analogous to the running likelihood-ratio \eqref{eq:running.lr} in that it makes use of lagged AERMs, but it does not require a correctly specified likelihood.

The online GUe-value requires computation of $n$-many AERMs and learning rates, which may be expensive. As an alternative, define the  \textit{offline} GUe-value
\begin{equation}
\label{eq:offline}
    G_{n, \text{off}}(\theta) \equiv G_{S, \text{off}}(\theta) := \exp\Bigl[ -\omega \, n_2 \{\rhat_{S_2}(\thetahat_{S_1}) - \rhat_{S_2}(\theta)\} \Bigr],
\end{equation}
where $S_1 \sqcup S_2$ is a partition of the sample $S$ into two sub-samples of size $n_1$ and $n_2$, respectively, $\thetahat_{S_i}$ is any AERM on $S_i$, and $\omega$ is a learning rate. Again, this is in analogy to the split likelihood-ratio of \citet{wasserman2020}. Because the online and offline GUe-values share many properties, we write $G_n$ when distinguishing between the two is unnecessary and refer to simply the GUe-value.  

The intuition for the GUe-value is that $G_n(\theta)$ is large if and only if a suitable empirical risk function at $\theta$ is large, suggesting that $\theta$ is highly inconsistent with the data compared to the estimators. It is also interesting to note that the offline GUe-value can be written as the ratio of Gibbs posterior probability density functions \citep[e.g.,][]{zhang2006,bissiri2016,grunwald2020,martin2022} when using the (possibly improper) uniform prior.  Hence, the offline GUe-value is like a ``generalized Bayes factor''---or a ``Gibbs factor''---between $\theta$ and the AERM $\thetahat_{S_1}$. These intuitions suggest that $G_n(\theta^*)$ should be rather small, and that only values $\theta$ with $G_n(\theta)$ sufficiently small ought to be considered plausible values for $\theta^*$. It turns out that this intuition is indeed sound, as we explain in the following subsection.

\subsection{Frequentist validity}
We should not refer to GUe-values as ``e-values'' or ``e-processes'' without first demonstrating that they satisfy the respective defining properties.  Unlike in the context of a well-specified statistical model, it is not possible to do this demonstration without imposing some conditions on the data-generating process $\D$ and the loss function $\ell$.  It turns out that the {\em strong central condition} advanced in \citet{vanErvan2015}, commonly used in the analysis of ERMs and Gibbs posteriors, is sufficient for our purposes as well.  

\begin{scc}
A learning problem determined by a data-generating process $\D$ on $\mathbb{Z}$ and a loss function $\ell:\Theta \times \mathbb{Z} \rightarrow \R^+$ satisfies the \textit{strong central condition} if there exists $\bar{\omega} > 0$ such that
\begin{equation*}
    {\E}_{Z\sim\D}\exp\bigl[-\omega\{ \ell(\theta; Z) - \ell(\theta^*; Z)\} \bigr] \leq 1 \quad \text{for all $\theta\in\Theta$ and all $\omega \in [0, \bar{\omega})$}.
\end{equation*}
In this case, we say that the condition holds with learning rate $\bar{\omega}$.
\end{scc}\label{cond:scc}

The strong central condition is effectively a bound on the moment generating function of $\ell(\theta^*; Z) - \ell(\theta; Z)$ in a small positive interval $[0,\bar\omega)$ containing the origin.  
As discussed in detail in \citet{vanErvan2015}, this condition holds in a number of practically relevant cases; see, also, \Cref{rem:sccsufficient} in \Cref{app:remarks} and \citet{grunwald2020}.  What motivates the strong central condition in the study of statistical learning via loss functions is that it is the ``right'' condition in a certain sense for these problems.  That is, almost all learning paradigms---including ERM, two-part minimum description length, and PAC-Bayes---achieve ``fast" rates (in the sense that the excess risk $R(\thetahat_n) - R(\theta^*)$ converges at a reasonably fast rate) only under this condition \citep{vanErvan2015}. As a concrete example, the commonly used Bernstein condition 
%(a generalization of the also-frequently-used Tsybakov margin condition) 
is equivalent to the strong central condition for bounded losses \citep{vanErvan2015}.

We should also emphasize that our main result shows the strong central condition is {\em sufficient} for validity of our proposal, no claims are made about the strong central condition being {\em necessary}.  Hence, our GUe-value may work even when the strong central condition fails.  See \Cref{re:no.scc} for some technical insights and \Cref{ex:heavytail} in \Cref{apdx:extrasims} for a numerical example showing that our GUe confidence intervals for the mean of a heavy-tailed distribution are empirically valid and more efficient than the provably valid intervals proposed in \citet{wang2023}.

One of our contributions is to provide a connection between the literature surrounding statistical learning and that of safe inference: The following two lemmas demonstrate that the strong central condition is sufficient to ensure that the online and offline GUe-values are e-processes and e-values, respectively.  

\begin{lemma}\label{lem:eprocess}
Under the strong central condition with learning rate $\bar{\omega}$, the online GUe-value $G_{n,\text{on}}(\theta^*)$ in \eqref{eq:online} is an e-process if $\omega \in [0, \bar{\omega})$. 
\end{lemma}

\begin{lemma}\label{lem:evalue}
Under the strong central condition with learning rate $\bar{\omega}$, the offline GUe-value $G_{n,\text{off}}(\theta^*)$ in \eqref{eq:offline} is an e-value if $\omega \in [0, \bar{\omega})$.
\end{lemma}

We can now begin to see the trade-off between the online and offline GUe-values: the online GUe-value is an e-process and hence has stronger error rate control properties, as described in our main result, \Cref{thm:evalue}.  The offline GUe-value is only an e-value, so its properties are generally weaker (e.g., combining offline GUe-values only maintains validity under optional continuation for independent offline GUe-values, whereas combining online GUe-values can maintain the anytime-valid property even if the online GUe-values are dependent), but it is typically less expensive to compute compared to the online variant that requires evaluation of the lagged AERMs.
 
\begin{theorem}\label{thm:evalue}
Suppose that the strong central condition holds with learning rate $\bar \omega$, and take $\omega \in [0, \bar\omega)$. Fix a desired significance level $\alpha \in (0,1)$. Then the test that rejects $H_0: \theta^* \in \Theta_0$ in favor of $H_1: \theta^* \not\in \Theta_0$ if and only if 
\[ G_n(\Theta_0) := \inf_{\theta \in \Theta_0} G_n(\theta) \geq \alpha^{-1}, \]
controls the frequentist Type~I error at level $\alpha$, i.e., 
\[ \Pr\{ G_n(\Theta_0) \geq \alpha^{-1} \} \leq \alpha, \quad \text{for all $\Theta_0$ that contain $\theta^*$}. \]
Also, the set estimator 
\[ C_\alpha(Z^n) := \bigl\{ \theta \in \Theta: G_n(\theta) < \alpha^{-1} \bigr\} \]
has frequentist coverage probability at least $1-\alpha$, i.e., 
\[ \Pr\{ C_\alpha(Z^n) \ni \theta^* \} \geq 1-\alpha. \]
Furthermore, for the online GUe-value specifically, the above tests and confidence sets are anytime-valid, i.e., for any stopping time $\tau$, 
\[ \Pr\{ G_\tau(\Theta_0) \geq \alpha^{-1} \} \leq \alpha \quad \text{and} \quad \Pr\{ C_\alpha(Z^\tau) \ni \theta^* \} \geq 1-\alpha.\]
\end{theorem}

Now the trade-off between the online and offline GUe-values is even more clear.  While both constructions lead to tests and confidence sets with finite-sample control of frequentist error rates, the online version is anytime-valid: the bounds hold uniformly over all stopping rules, but generally with a higher computational cost. The advantage of anytime-validity, again, is that the method is robust to the common practice of making within-study decisions about whether to proceed with further data collection and analysis.  

\section{Learning the learning rate}\label{sec:learning_rate}

\subsection{Adaptive GUe and the unit-dominance condition}

The learning rate $\omega$ is critical for the validity and efficiency of the GUe-value hypothesis tests and confidence sets. On the one hand, if $\omega$ is too large, then $G_n(\theta^*)$ is smaller than it should be, the confidence sets are likewise too small and hence are likely to undercover.  On the other hand, if $\omega$ is too small, the confidence sets for $\theta^*$ are larger than necessary, resulting in inference that is overly conservative. 

Because it is generally impossible to know the ``correct" learning {\em a priori}, its value in practical applications must be chosen in a data-driven manner. The proofs of \Cref{lem:eprocess,lem:evalue}, however, do not accommodate data-driven choices of $\omega$. Hence, we propose the following adaptive version of the GUe-value for those practical setting where a value of $\omega$ is learned empirically:
\begin{equation*}
    \widehat{G}_{n,\text{on}}(\theta) := \exp\Bigl[\sum_{i=1}^n -\omegahat_i \cdot \{ \ell(\thetahat_{i-1}; Z_i) - \ell(\theta; Z_i) \} \Bigr]
\end{equation*}
and
\begin{equation*}
        \widehat{G}_{n, \text{off}}(\theta) := \exp\Bigl[ -\omegahat_{S_1}\cdot n_2 \{\rhat_{S_2}(\thetahat_{S_1}) - \rhat_{S_2}(\theta)\} \Bigr],
\end{equation*}
where $\omegahat_k$ may depend on the first $k$ data points $Z_1, \ldots, Z_k$, and $\omegahat_{S_1}$ may depend on the training set $S_1$. With these new definitions that allow the learning rate to depend on observed data, we offer the following sufficient condition for validity:
\begin{cc}
Let $\D$ be a data-generating distribution over a set $\mathbb{Z}$ and write the learning rate and AERM as functions $\omegahat:\mathbb{Z}^\infty \rightarrow \R^+$ and $\thetahat:\mathbb{Z}^\infty \rightarrow \Theta$ of the data, respectively. We say that the unit-dominance condition holds for $(\omegahat, \thetahat, \D)$ if 
    \begin{equation*}
        \E\qty[\exp(-\omegahat(Z^n) \cdot \{\ell(\thetahat(Z^{n-1})\altgiven Z_n) - \ell(\theta^*\altgiven Z_n)\}) \given Z^{n-1}] \leq 1   \enspace \text{almost surely}
    \end{equation*}
    for every $n\in\N$.
\end{cc}

It is easy to verify that if the strong central condition holds with learning rate $\overline\omega$ and if the data-dependent learning rates $\omegahat_n$ are such that $\omegahat_n \leq \overline\omega$ almost surely, then the unit-dominance condition also holds.  
%straightforward to see that if the learning rates are chosen such that the strong central condition holds, the compatibility condition also automatically holds; 
In this sense, the unit-dominance condition is weaker than the strong central condition to fulfill, but arguably the former is rather difficult to verify.  Nevertheless, 
%(though in turn, it is more difficult to know whether or not the compatibility condition actually holds). Despite being weaker, 
this condition does allow the proofs of \Cref{lem:eprocess,lem:evalue} to hold even with data-driven choices of learning rates using the exact same proof strategies. We thus have the following theorem:
\begin{theorem}
    If the unit-dominance condition holds, then $\Ghat_{n, \text{on}}$ is an e-process and $\Ghat_{n, \text{off}}$ is an e-value. Consequently, when the unit-dominance condition holds, the test that rejects $H_0: \theta^*\in\Theta_0$ in favor of $H_1: \theta^* \not\in\Theta_0$ if and only if $\Ghat_n(\Theta_0) \geq \alpha^{-1}$ controls the frequentist Type I error at level $\alpha$. For $\Ghat_{n, \text{on}}$ specifically, this test is anytime-valid. 
\end{theorem}

%\subsection{Power of a Data-Driven Learning Rate}
\subsection{Power of adaptive GUe}

Our validity results do not rely specifically on $\thetahat_{i-1}$ and $\thetahat_{S_1}$ being AERMs.  Our primary motivation for choosing AERMs is for the sake of efficiency: AERMs are often consistent estimators of the risk minimizer, and this property leads to analogous large-sample consistency results for the above tests and confidence regions.  The next two theorems present successively stronger results along these lines. 

\begin{theorem}\label{thm:power}
Suppose $\theta$ is such that $R(\theta) > \inf_{\vartheta\in\Theta} R(\vartheta)$. Further suppose that $\sup_{\theta\in\Theta} |\rhat_n(\theta) - R(\theta)|  \ip 0$ as $n\rightarrow\infty$.
\begin{enumerate}
    \item If there exists a constant $\underline{\Omega} > 0$ such that 
    %the learning rates satisfy 
    $\liminf\limits_{n_1\rightarrow\infty} \omegahat_{S_1} \geq \underline{\Omega}$ almost surely, then $$\lim_{n \rightarrow \infty} \Pr\bigl\{ \widehat{G}_{n, \text{off}}(\theta) \geq \alpha^{-1} \bigr\} = 1$$
    \item Suppose that there exist positive constants $\overline{\Omega}$ and $\underline{\Omega}$ such that the learning rates satisfy $\omegahat_n\leq \overline{\Omega}$ almost surely for all $n$ and $\limsup_n \omegahat_n \geq \underline{\Omega}$ almost surely. If the $(\varepsilon, \delta)$-AERM mapping $z^k \mapsto \thetahat(z^k)$ is leave-one-out stable in the sense that 
    \begin{equation}
    \label{eq:estimator.stability}
    |\ell\{\thetahat(z^{n-1}); z_i\} - \ell\{\thetahat(z^{n}); z_i\}| = o(1), \quad \text{for all $(z_1,z_2,\ldots) \in \mathbb{Z}^\infty$ and all $i\in\{1, \ldots, n\}$}, 
    \end{equation}
    and either $\varepsilon > 0$ or $\delta < \underline{\Omega} \cdot \overline{\Omega}^{-1} \left\{R(\theta) - \inf_{\vartheta\in\Theta} R(\vartheta)\right\}/7$, then $$\lim_{n \rightarrow \infty} \Pr\bigl\{ \widehat{G}_{n, \text{on}}(\theta) \geq \alpha^{-1} \bigr\} = 1$$
\end{enumerate}
\end{theorem}

\Cref{thm:power} says that, under regularity conditions, the power function for the GUe-value test of the point null $H_0: \theta^* = \theta$ converges to $1$ for each $\theta$ that is not a risk minimizer. Hence, for any non-risk minimizing $\theta$, we see that the associated $(1-\alpha)$-level confidence set for $\theta^*$ shrinks to eventually exclude $\theta$ with high probability as more data is collected. The uniform convergence of the empirical risk to the risk is a standard condition, as it is sufficient for $\rhat_n(\thetahat_n) \rightarrow R(\theta^*)$ in the first place. Similarly, it is typical to require some form of estimator stability in the online setting in order to learn $\theta^*$ \citep[e.g.,][]{bousquet2002, rakhlin2005, shalevschwartz2010}. For example, the ERM is leave-one-out stable in the sense of \eqref{eq:estimator.stability} [with rate $n^{-1}$, see \eqref{eq:estimator.stability.rate}] if the loss function is smooth and strongly convex over Euclidean space \citep[Theorem 7.10]{zhang2023}.

\begin{remark}
\label{re:no.scc}
{\em We also note that the above theorem holds for any positive choice of learning rate $\omega$. That is, even if $\omega$ is set to be too large or, otherwise, the strong central condition fails, the adaptive GUe-value still has desirable power. That said, we caution that powerful tests without finite-sample validity guarantees must be used with care.}
\end{remark}

\begin{remark}\normalfont{
It might seem surprising that the power result for the adaptive online GUe requires $\limsup_n \omegahat_n \geq \underline{\Omega}$ while the same result for the adaptive offline GUe requires the stronger condition $\liminf_{n_1} \omegahat_{S_1} \geq \underline{\Omega}$. We believe, however, that the stronger condition cannot be avoided for the offline GUe-value. This is because the adaptive online GUe-value ``remembers" its previous values, so as long as one occasionally chooses nonzero learning rates, $\Ghat_{n, \text{on}}$ continues growing. In contrast, the adaptive offline GUe-value only has access to the learning rate it chooses for the particular sample size, so alternating $\omegahat_{S_1}$ between a zero and nonzero value would prevent the relevant limit from existing. Furthermore, even restricting $\omegahat_{S_1}$ to nonzero values, there are no other safeguards in the hypotheses preventing an adverserial selection strategy from choosing $S_2$ and $\omegahat_{S_1}$ in such a way that $\Ghat_n$ remains small infinitely often.}
\end{remark}

\begin{remark}\normalfont{
    The sufficient condition for unit-dominance, i.e., $\omegahat_n \leq \overline{\omega}$, might appear to be in tension with the above theorem's requirement $\limsup \omegahat_n \geq \underline{\Omega}$, but that is not the case.  Whereas $\overline{\omega}$ from the strong central condition is fixed and problem-specific, the constant $\underline{\Omega}$ may be arbitrarily small and is (indirectly) chosen by the practitioner via the choice of learning-rate selection algorithm. Consequently, so long as $\underline{\Omega} < \overline{\omega}$, we may enjoy both frequentist validity and desirable power properties; see \Cref{fig:powerexample}.
}\end{remark}

Vanishing Type~II error probability under fixed alternatives is a relatively weak property.  A more refined analysis considers alternatives $\theta_n$ that are different from but converging to the risk minimizer.  Then the relevant question is: How fast can the alternative $\theta_n$ converge to $\theta^*$ and still the adaptive GUe-value can distinguish the two? The following theorem gives an answer to this question, effectively bounding the radius of the adaptive GUe-value confidence set.  That is, for $\beta$ as defined below, the confidence set contains a point that has risk $\gtrsim n^{-\beta}$ more than $\theta^*$ with vanishing probability.

\begin{theorem}\label{thm:radius}
Fix $\beta\in (0, 1)$ and let $(\theta_n)_{n\in\N}$ be a sequence in $\Theta$ such that $R(\theta_n) - \inf_\vartheta R(\vartheta) \gtrsim n^{-\beta}$. Then for any $\omega > 0$, we have the following rate results for the adaptive GUe-value:
\begin{enumerate}
    \item Suppose that $\sup_\theta |\rhat_n(\theta) - R(\theta)| = o_p(n^{-\beta})$. Further suppose that there exist positive constants $\overline{\Omega}$ and $\underline{\Omega}$ such that the learning rates satisfy $\omegahat_n\leq \overline{\Omega}$ almost surely for all $n$ and $\limsup_n \omegahat_n \geq \underline{\Omega}$ almost surely. If the $(\varepsilon, \delta)$-AERM mapping $z^k \mapsto \thetahat(z^k)$ is leave-one-out stable at rate $n^{-\beta}$, i.e., 
    \begin{equation}
    \label{eq:estimator.stability.rate}
        |\ell\{\thetahat(z^{n-1}); z_n\} - \ell\{\thetahat(z^{n}); z_n\}| = o(n^{-\beta}), \quad \text{for all $(z_1,z_2,\ldots) \in \mathbb{Z}^\infty$}, 
    \end{equation}
    and $\varepsilon > 0$, then $\lim_{n\rightarrow\infty} \Pr\bigl\{ \Ghat_{n, \text{on}}(\theta_n) \geq \alpha^{-1} \bigr\} = 1$. 
    % \begin{equation*}
    %     \lim_{n\rightarrow\infty} \Pr\bigl\{ G_{n, \text{on}}^{(\omega)}(\theta_n) \geq \alpha^{-1} \bigr\} = 1.
    % \end{equation*}
    \item Suppose that  $\sup_\theta |\rhat_{S_2}(\theta) - R(\theta)| = o_p(n_2^{-\beta})$, and that there exists a positive constant $\underline{\Omega}$ such that the learning rates satisfy $\liminf\limits_{n_1\rightarrow\infty} \omegahat_{S_1} \geq \underline{\Omega}$ almost surely. If $R(\thetahat_{S_1}) \ip \inf_\vartheta R(\vartheta)$ as $n_1\rightarrow \infty$, then 
    \begin{equation*}
        \lim_{(n_1, n_2)\rightarrow(\infty, \infty)} \Pr\bigl\{ \Ghat_{n, \text{off}}(\theta_{n_2}) \geq \alpha^{-1} \bigr\} = 1.
    \end{equation*}
    \item Suppose that $\sup_{\theta} |\rhat_{S_2}(\theta) - R(\theta)| = o_p(n_2^{-\beta})$, and that there exists a positive constant $\underline{\Omega}$ such that the learning rates satisfy $\liminf\limits_{n_1\rightarrow\infty} \omegahat_{S_1} \geq \underline{\Omega}$ almost surely. If $R(\thetahat_{S_1}) \ip \inf_\vartheta R(\vartheta)$ as $n_1\rightarrow \infty$ and $n_1 \lesssim n_2$, then 
\begin{equation*}
    \lim_{(n_1, n_2)\rightarrow(\infty, \infty)} \Pr\left\{\Ghat_{S, \text{off}}(\theta_{n}) \geq \alpha^{-1}\right\} = 1.
\end{equation*}
\end{enumerate} 
\end{theorem}

Note that $R(\thetahat_{S_1}) \ip \inf_\vartheta R(\vartheta)$ as $n_1 \rightarrow \infty$ holds if $S_1$ is identically distributed to $S_2$, and that $n_1 \lesssim n_2$ holds if the sample splitting occurs in a constant proportion. Then \Cref{thm:radius} says that for the power to exhibit desirable behavior, it only requires uniform convergence of the empirical risk at a typically assumed rate (and, in the offline case, for our split between training and validation sub-samples to be done at random). Furthermore, the proof of this theorem illustrates that the size of our confidence set decays at a rate no faster than the rate that the estimator converges to the infimum risk; for this reason, it is preferable to use ERMs as the estimator as opposed to $\delta$-AERMs for large $\delta$. 

The rate requirement of \Cref{thm:radius} is far from restrictive: a rate of about $o_p(n^{-1/2})$ is fairly typical.
As a concrete example, suppose $(X_1, Y_1), \ldots, (X_n, Y_n)$ are i.i.d.~random vectors from any distribution $\D$ over $\X \times \{0, 1\}$, let $h:\X \times \Theta \rightarrow \{0, 1\}$ be any measurable function, and 
consider the zero-one loss function $\ell\{\theta\altgiven (x, y)\} = \mathbb{I}\{y \neq h(x\altgiven \theta)\}$. If the set $\{x\mapsto h(x\altgiven\theta) \mid \theta\in\Theta\}$ is of finite VC dimension (e.g., the set is a subset of a finite-dimensional affine space), it follows from Corollary 3 of \citet{hanneke2016} that $\sup_{\theta}|\rhat_n(\theta) - R(\theta)| = o_p(n^{-\beta})$ for any $\beta < 1/2$, so we see that the rate condition of the theorem holds. Indeed, Theorem 17.1 of \citet{anthony2009} implies that this rate of $o_p(n^{-\beta})$ for any $\beta < 1/2$ holds for \textit{any} bounded loss function of finite fat-shattering dimension.

Note that the above theorem may be sub-optimal for unbounded losses. For example, one may hope for mean estimation via $L^2$ loss to have confidence sets of size $O_P(n^{-1/2 + \delta})$ for every $\delta > 0$. This is not provided by directly using \Cref{thm:radius} since, to achieve this, the theorem's hypotheses would require $\sup_\theta|\rhat_n(\theta) - R(\theta)| = o_p(n^{-1 + 2\delta})$, which is generally out of reach. However, this weakness is only an artifact of the conclusions of \Cref{thm:radius} being required to hold for \textit{all} loss functions; for the specific case of the $L^2$ loss, \Cref{sec:l2power} shows that the adaptive GUe confidence sets do have size $O_P(n^{-1/2 + \delta})$ for all $\delta > 0$---in fact, the size is $O_P(n^{-1/2} \log n)$.

Finally, although confidence sets only make sense when the risk minimizer $\theta^*$ exists, \Cref{thm:power,thm:radius} apply even if $\inf_\vartheta R(\vartheta)$ is not attained.  See \Cref{rem:infnotexist} in \Cref{app:remarks} for further details on this.

\subsection{Learning rate insensitivity}
An important open question concerns the practical choice of learning rate in the adaptive GUe-value.  Fortunately, as the following two numerical examples show, the adaptive GUe-value is largely insensitive to the choice of learning rate sequence---that is, virtually any reasonable choice of learning rate ought to suffice.

\begin{example}\normalfont{
Consider the $K$-means algorithm, an unsupervised learning method that clusters data $Z_1, \ldots, Z_n$ into $K$ clusters, with $K$ fixed in advance, where each cluster has minimum within-cluster variance. Specifically, $K$-means aims to find a partition $\theta = (\theta_1,\ldots,\theta_K)$ of the data, where $\theta_k \subseteq \{1,\ldots,n\}$ for each $k=1,\ldots,K$, that minimizes
\begin{equation*}
\ell(\theta; Z^n) = \sum_{k=1}^K |\theta_k| \, \operatorname{Var}(\{Z_i: i \in \theta_k\}), 
\end{equation*}
where $|A|$ denotes the cardinality of the set $A$ and the $\operatorname{Var}$ operator returns the sample variance of its data-set-valued argument; note that we can define $\operatorname{Var}(\varnothing)$ arbitrarily, here, since multiplying by the cardinality (of $\varnothing$) eliminates the dependence on this arbitrary choice. This partition $\theta$ implicitly defines the centroids $\mu_1, \ldots, \mu_K$, where $\mu_k = |\theta_k|^{-1} \sum_{i\in \theta_k} z_i$.  These centroids are typically the quantities of interest. 

We generate bivariate normal data from $K=3$ populations, $\operatorname{N}_2(\mu_k, \sigma^2 I)$, where $\sigma^2=0.01$ and $\mu_1 = (1,0)^\top$, $\mu_2 = (-1/2, \sqrt{3}/2)^\top$, and $\mu_3 = (-1/2, -\sqrt3/2)^\top$.  Then the true centroids for $K$-means with $K = 3$ are approximately the means of each population.

\begin{figure}[t]
    \centering
    \includegraphics[width=0.45\textwidth]{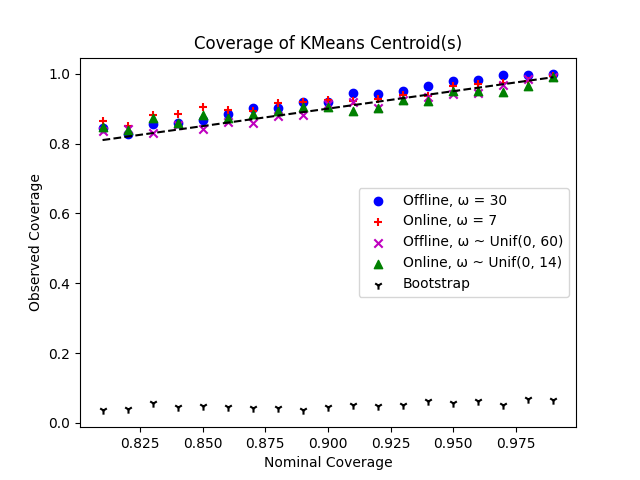}
    \caption{Coverage of $\mu_3$ for the bootstrapped confidence set versus joint coverage of $(\mu_1, \mu_2, \mu_3)$ for the GUe confidence sets.}
    \label{fig:kmeans}
\end{figure}

One commonly-used method to construct approximate confidence sets for these centroids is via bootstrapping \citep{hofmans2015}. That is, one resamples from the observed data set that has the estimate $\widehat{\mu}$ for the centroid, performs $K$-means again on the bootstrapped data, and creates an ellipse with major and minor axes based on the covariance matrix necessary for the ellipse to contain the $\widehat{\mu}$ with the nominal level of coverage over the bootstrap resamples.  Here, we compare this procedure for uncertainty quantification about the $K$-means centroids to our proposed generalized universal inference framework.  

If we draw 100 samples with equal probability from the three populations, then these bootstrapped confidence sets indeed attain approximately the correct nominal coverage. However, this changes when the populations are unbalanced. \Cref{fig:kmeans} illustrates that when the populations are sampled from with probabilities $0.96$, $0.03$, and $0.01$, respectively, bootstrapping leads to abysmal coverage for the least frequent population centroid (and thus would perform even worse if it were used to create a joint confidence set for all three centroids). In contrast, it is seen that both the online and offline GUe-values (with $\omega = 7$ and $\omega = 30$ respectively) possess essentially the correct level of coverage for the joint vector $(\mu_1, \mu_2, \mu_3)$. Furthermore, it is seen that if the learning rates are chosen \textit{uniformly at random} within 100\% of the ``correct" values, i.e., $\omegahat_{S_1} \sim \operatorname{Unif}(0, 60)$ and each $\omegahat_{i} \overset{\text{i.i.d.}}{\sim} \operatorname{Unif}(0, 14)$, then the GUe-value still attains the correct level of coverage.
}\end{example}

\begin{example}\normalfont{
Consider the problem of binary classification in $\mathbb{R}$ via a support vector machine, in which given data $(x_i, y_i)_{i=1}^n$ with $x_i \in \R$ and $y_i \in \{-1, 1\}$, one wishes to find $\theta = (\theta_1, \theta_2)$ minimizing
\begin{equation*}
    \ell(\theta\altgiven X^n, Y^n) = \frac{1}{n}\sum_{i=1}^n \max\{0, 1 - Y_i(\theta_1 + \theta_2 X_i)\}.
\end{equation*}
We generate data from the logistic model $\Pr(Y = 1 \given X) = \operatorname{expit}(X)$ with $X\sim \operatorname{N}(0,1)$. Furthermore, to illustrate the anytime validity aspect, rather than collecting a sample of fixed size, we collect data until the stopping rule $\sum x_i^2 > 10$ is satisfied. 
\begin{figure}
    \centering
    \includegraphics[width=0.45\linewidth]{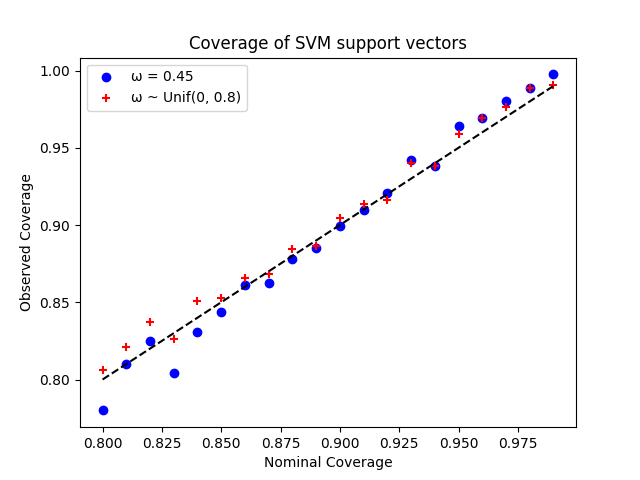}
    \caption{Coverage of optimal support vector from SVM by the online GUe-value}
    \label{fig:svm}
\end{figure}
\Cref{fig:svm} demonstrates that once again, by selecting the ``correct" choice of $\omega = 0.45$ as the learning rate for the online GUe-value, one can obtain essentially the correct level of coverage. Furthermore, even by selecting the sequence of learning rates uniformly at random in the interval $[0, 0.8]$, one still obtains approximately correct levels of coverage. 
}
\end{example}

\subsection{Algorithmic selection of learning rates}
The previous two examples show that selecting learning rates completely at random within a fairly large neighborhood of the ``correct'' value still empirically satisfies validity and anytime-validity.  Evidently, the main focus of any learning rate selection algorithm should simply be to land within this neighborhood. We previously noted that the offline GUe-value can be written as a ratio of Gibbs posterior densities, so it is not unreasonable to apply a learning rate selection method used to choose a Gibbs posterior learning rate. Some proposals for the latter include the unit information loss approach \citep{bissiri2016}, matching information gain \citep{holmes2017}, the asymptotic Fisher information matching approach \citep{lyddon2019}, the R-Safe Bayes algorithm \citep{grunwald2017}, and a sample-splitting strategy \citep{perrotta2020}, among others. \citet{wu2023} found that, with the learning rate chosen according to these strategies, the corresponding Gibbs posterior credible sets generally fail to achieve the nominal frequentist coverage probability.  They do, however, identify one algorithm that maintains valid frequentist coverage even under model misspecification: the general posterior calibration (GPC) algorithm of \citet{syring2019}; see, also, \citet{martin2022}. For further discussion on our choice of GPC for selecting the learning rate, see \Cref{rem:safebayes} in \Cref{app:remarks}. 

The GPC algorithm proceeds by constructing a $1-\alpha$ level confidence set for $\theta^*$ using the nonparametric bootstrap, resampling from the original sample $S$ and choosing $\omega$ such that the credible set contains $\thetahat_S$ with probability $1-\alpha$ over the bootstrap resamples. Given that the GPC algorithm does well in attaining valid confidence sets from the Gibbs posterior, it is sensible to similarly use the nonparametric bootstrap to select the learning rate for the
GUe-value. This approach is detailed in \Cref{alg:nonparam}. 

\begin{algorithm}[t]
    \caption{Nonparametric Bootstrap for Learning Rate Calibration}\label{alg:nonparam}
    \begin{algorithmic}
        \Require $(z_1, \ldots, z_n)$, a dataset we may train on
        \Require $\Omega$, a set of candidate learning rates
        \Require $\alpha$, a significance level to calibrate to
        \Require $N$, the number of bootstrap iterations to do

        \State Compute $\thetahat$, the ERM for $(z_{1}, \ldots, z_{n})$
        \State $\text{coverages}(\omega) \gets 0\textbf{ for all } \omega\in\Omega$
        \For{$\omega\in\Omega$}
            \For{$i$ in $1,\ldots, N$}
                \State Draw $S_B = (z_{b(1)}, \ldots, z_{b(n)})$ uniformly from $(z_1, \ldots, z_n)$
                \If{$\widehat{G}_{S_B}(\thetahat) < 1/\alpha$, where every $\omegahat = \omega$}
                    \State $\text{coverages}(\omega) \gets \text{coverages}(\omega) + 1/N$
                \EndIf
            \EndFor
        \EndFor
        \State \Return $\argmin_{\omega\in\Omega} |\text{coverages}(\omega) - (1-\alpha)|$
    \end{algorithmic}
\end{algorithm}

\begin{example}\normalfont{
To demonstrate empirically that \Cref{alg:nonparam} does indeed satisfy the unit-dominance criterion, hence yielding anytime-validity, we simulate independent sequences of i.i.d.~standard exponentially distributed data, and estimate  
\begin{equation}\label{eq:compatcondvalue}
    \E\qty[\exp(-\omegahat_n \cdot \{\ell(\thetahat_{n-1}\altgiven Z_i) - \ell(\theta^*\altgiven Z_i)\}) \given Z^{n-1}]
\end{equation}
for each $n \in \{1, 2, \ldots, 100\}$, when using the $L^1$ loss and using \Cref{alg:nonparam} (with $\alpha = 0.05$) to select each $\omegahat_n$. \Cref{fig:compatcond} shows the estimated value of (\ref{eq:compatcondvalue}) over 6 different samples of size 100, as well as the lower bound for the joint 95\% confidence interval for (\ref{eq:compatcondvalue})---as we recall that the unit-dominance condition requires that (\ref{eq:compatcondvalue}) be at most $1$ for all $n\in\mathbb{N}$. We see that in the visualized examples, we do indeed estimate (\ref{eq:compatcondvalue}) to be at most $1$ in all cases; indeed, over 100 independent samples, all of the samples have the 95\% joint confidence interval for (\ref{eq:compatcondvalue}) contain $1$ for all $n\in\{1, \ldots, 100\}$, giving strong evidence that \Cref{alg:nonparam} does indeed satisfy the unit-dominance condition.  

In addition to validity, we may also verify in this example that \Cref{alg:nonparam} leads to GUe-value-based tests with reasonable power properties. First, \Cref{fig:powerexample} (left panel) empirically demonstrates that the chosen learning rates satisfy $\limsup_n \omegahat_n > 0$ as required by \Cref{thm:power}.  Additional simulations (not shown) confirm that the overall pattern in the left panel of \Cref{fig:powerexample} does not change significantly when the underlying distribution changes.  Indeed, the learning rate distribution plots look roughly the same as in \Cref{fig:powerexample} when the data are shifted exponentials with different medians, Poissons, (discrete or continuous) uniforms, etc. Second, we consider testing $H_0: \text{median} = \log(2)$ versus $H_1: \text{median} \neq \log(2)$, where $\log(2)$ corresponds to the median of the standard exponential distribution.  Then \Cref{fig:powerexample} (right panel) shows the vanishing Type~II error rates for the test based on three different alternative distributions with medians different from $\log(2)$.
}
\end{example}

\begin{figure}[t]
    \centering
    \includegraphics[width=0.9\textwidth]{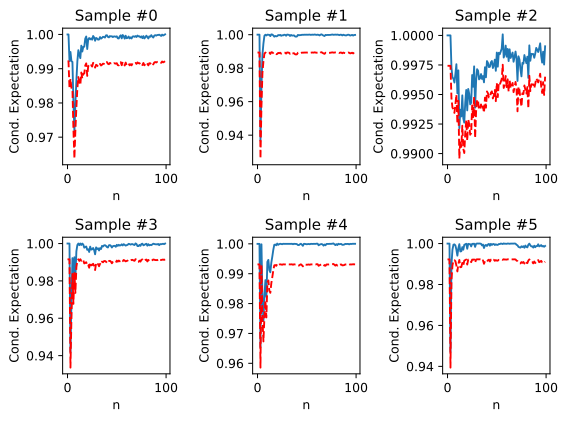}
    \caption{Estimates (in blue) for the quantity for the unit-dominance condition on i.i.d.~$\operatorname{Exponential}(1)$ samples of size 100, as well as lower bounds (in red, dashed) for the 95\% joint confidence interval for this quantity over $n \in \{1, \ldots, 100\}.$}
    \label{fig:compatcond}
\end{figure}

\begin{figure}[t]
    \centering
    \includegraphics[width=0.45\textwidth]{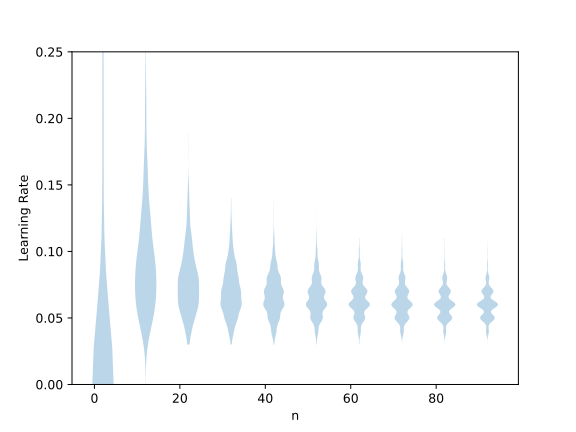}
    \includegraphics[width=0.45\textwidth]{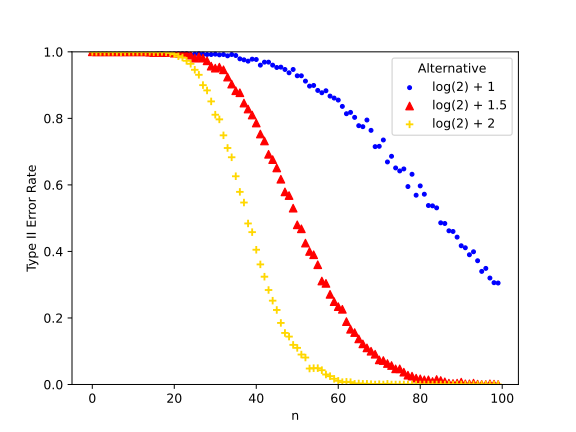}
    \caption{The distribution of learning rates chosen by \Cref{alg:nonparam} versus sample size (left) and the Type II error rate for testing $H_0: \text{median} = \log(2)$ versus $H_1: \text{median} \neq \log(2)$ for three different alternative distributions having medians as stated in the legend (right).}
    \label{fig:powerexample}
\end{figure}

Because the nonparametric bootstrap chooses an appropriate learning rate in a principled manner agnostic to how the data are distributed, it still tends to be conservative in general, choosing smaller learning rates than necessary (though it is noteworthy that universal inference is known to be conservative in general regardless due to application of Markov's inequality; see \citealt{park2023}). If one is sure that the data come from a particular parametric model, one may obtain less conservative choices of learning rates by instead employing the parametric bootstrap to choose a learning rate---at the expense of possibly having confidence sets with below-nominal-level coverage if the model is actually misspecified; see \Cref{sec:paramboot}.

One can also obtain exactly correct choices for the learning rate under certain distributional assumptions and for certain loss functions. \Cref{sec:l2lranlys} presents such results for the $L^2$ loss function for the mean of a random variable, subject to the strong central condition. In particular, we demonstrate that for Gaussian distributed data, one can theoretically calculate a learning rate for the GUe confidence set that obtains exactly the correct coverage; furthermore, this learning rate asymptotically yields the correct coverage rate even for non-Gaussian data, due to the central limit theorem.  We also discuss the existence of a learning rate that obtains at least the nominal coverage if the data are indeed i.i.d.~and one either knows or has good estimates for the second and third moments of the population. These results are admittedly narrow in scope, either requiring strong distributional assumptions or once again relying on asymptotics rather than guaranteeing finite-sample validity. However, we expect that the construction of efficient e-values for complex problems will inevitably require some sort of data-driven tuning, so our results provide a useful starting point for these developments.

\section{Simulation studies}\label{sec:simulation}
\subsection{Bounded mean estimation}\label{sec:compare}
Another methodology that shares similar aims as our GUe confidence sets is given by the \textit{predictable plugin empirical Bernstein} (PrPl-EB) confidence sets of \citet{waudbysmith2023}, which are also safe confidence sets due to e-process properties, but are limited to estimating the mean of a bounded random variable. Given a sample $Z_1, \ldots, Z_n$, the $1-\alpha$ level PrPl-EB confidence interval is given by $\bigcap_{t=1}^n C_t$, where 
\begin{align*}
    C_t &:= \qty(\frac{\sum_{i=1}^t\lambda_i Z_i}{\sum_{i=1}^t \lambda_i} \pm \frac{\log(2/\alpha)+\sum_{i=1}^t (Z_i-\widehat{\mu}_{i-1})^2(-\log(1-\lambda_i)-\lambda_i)}{\sum_{i=1}^t \lambda_i}) \\
    \lambda_t &:= \min\qty(c, \sqrt{\frac{2\log(2/\alpha)}{\widehat{\sigma}^2_{t-1} t \log(1+t)}}) \\
    \widehat{\sigma}_t^2 &:= \frac{1/4 + \sum_{i=1}^t (Z_i - \widehat{\mu}_i)^2}{t+1} \\
    \widehat{\mu}_t &:= \frac{1/2 + \sum_{i=1}^t Z_i}{t+1},
\end{align*}
and $c$ is any reasonable value in $(0, 1)$---we follow the authors' recommendation of $c = 1/2$. The width of the PrPl-EB confidence interval in the i.i.d.~setting scales with the true (unknown) standard deviation, and thus obtains reasonable coverages at large samples. However, as Figure 20 of \citet{waudbysmith2023} shows, for modest sample sizes, the PrPl-EB confidence interval tends to cover almost the entirety of the support.

\Cref{fig:ramdas_comparison} compares the coverage of the PrPl-EB confidence set and the GUe confidence sets on an i.i.d.~sample of size $10$ from the beta distribution, where the learning rate for the GUe confidence sets were chosen via \Cref{alg:nonparam}. In agreement with the findings of \citet{waudbysmith2023}, the PrPl-EB confidence set covers the entire interval $[0, 1]$ at such a small sample size; the GUe confidence sets, however, are more efficient. In fact, both the online and offline GUe confidence sets attain approximately the correct coverage for $\alpha < 0.05$. \Cref{fig:ramdas_eq1} presents the coverage of the offline GUe confidence set when using the learning rate suggested by \Cref{prop:normallr} in \Cref{sec:l2lranlys} (i.e., the learning rate that yields asymptotically correct coverage), as well as when dividing this learning rate by 2 (which is our suggestion to ensure correct coverage at finite sample sizes). Although \Cref{prop:normallr} requires either Gaussian data or a large enough sample size for the central limit theorem to apply (neither of which is true in this case), our GUe confidence sets are still approximately calibrated---modulo some mild undercoverage. On the other hand, our suggested heuristic of halving this learning rate is more than conservative enough to hit the nominal coverage in these examples.

\begin{figure}[t]
    \centering
    \includegraphics[width=0.45\textwidth]{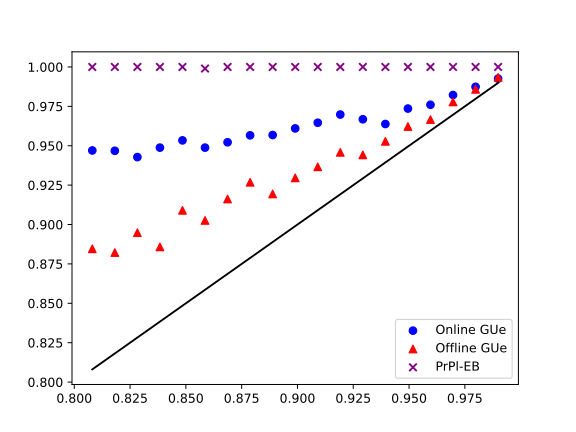}
    \caption{Nominal versus observed coverage of the PrPl-EB and GUe confidence sets based on i.i.d. $\operatorname{Beta}(5, 2)$ data.}
    \label{fig:ramdas_comparison}
\end{figure}

\begin{figure}[t]
    \centering
    \includegraphics[width=0.45\textwidth]{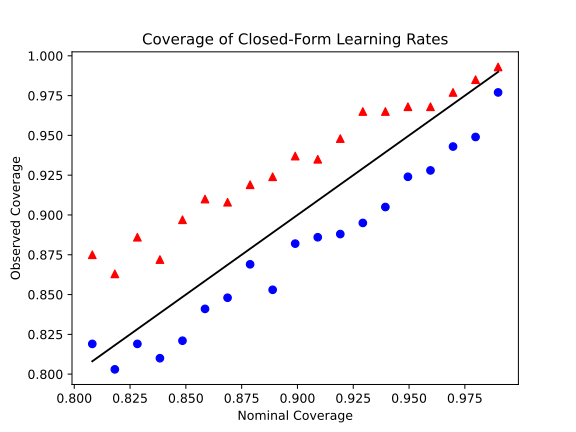}
    \caption{Nominal versus observed coverage of the offline GUe confidence sets based on i.i.d. $\operatorname{Beta}(5, 2)$ data. Blue circles indicate coverages when the learning rate is taken from \Cref{prop:normallr} in \Cref{sec:l2lranlys}, and red triangles are those when the learning rate is taken as half the value from \Cref{prop:normallr}}
    \label{fig:ramdas_eq1}
\end{figure}

\subsection{Replication crisis-related applications}
The replication crisis in science is a problem that has received significant attention in recent years. In this subsection, we showcase a variety of common problems that facilitate the lack of replicability of scientific experiments, and we demonstrate how these problems are mitigated by our GUe-value proposals.

\begin{example}\label{ex:falsenull}\normalfont{
Consider the following simple setup: A scientist is studying two populations that are distributed on $\R$ and wants to find the best threshold separating the two populations. That is, given data $(X_1, Y_1), \ldots, (X_n, Y_n)$ where $X_i\in\R$ are the observed data and $Y_i \in \{0, 1\}$ are the labels indicating which population the corresponding $X_i$ belong to, the scientist wishes to find the risk minimizer corresponding to the loss function
\begin{equation*}
    \ell(\theta\altgiven X, Y) = \mathbbm{1}(X\leq \theta)\cdot \mathbbm{1}(Y = 1) + \mathbbm{1}(X > \theta)\cdot \mathbbm{1}(Y = 0).
\end{equation*}
If the scientist does everything by the book---collecting a single data set of independent observations of a fixed, predetermined sample size from a known distribution, then generates a confidence interval from this data---then it is no surprise that the confidence interval works as planned: For any nominal coverage level, the practitioner shall observe precisely that level of coverage. This is not always what happens in applications, however. What often occurs is that the scientist has a null hypothesis $H_0: \theta^* = 0$ and an alternative hypothesis $H_1: \theta^* \neq 0$, and funding or publication hinges on the null hypothesis being rejected. Thus, especially when gathering data is expensive, the scientist may be tempted to gather more data when the data set collected so far fails to reject the null, and then stop collecting data once the null is rejected. Figure~\ref{fig:exact_stopping} demonstrates the effects of such a stopping rule.  The classical confidence intervals generated by the scientist tend to be less than the nominal coverage. On the other hand, 
%thanks to the anytime-validity property of e-processes, 
the online GUe confidence sets with the learning rate chosen via the nonparametric bootstrap exhibit coverage at or above the nominal level (n.b., the deviations below the nominal level are well within Monte Carlo error). Moreover and quite interestingly, even though the offline GUe-value is not provably an e-process, it too exhibits coverage approximately at the nominal level.

\begin{figure}[t]
    \centering
    \includegraphics[width=0.45\textwidth]{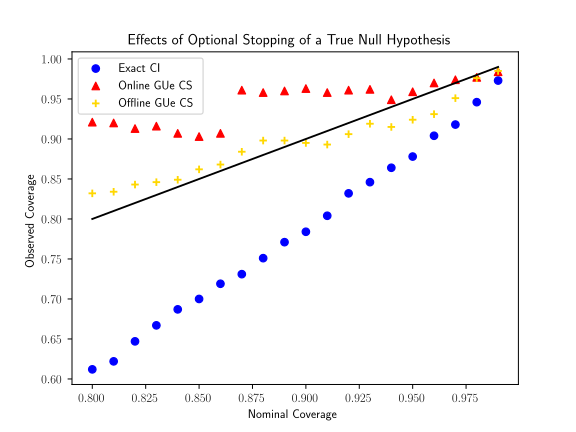}
    \caption{Nominal versus observed coverage of the ``exact" and GUe confidence sets when collecting data until the null hypothesis $H_0: \theta^* = 0$ is rejected for $\frac{1}{2}\operatorname{N}(\mu_1, \sigma^2) + \frac{1}{2}\operatorname{N}(\mu_2, \sigma^2)$ data, with $\mu_1 = 5$, $\mu_2 = 10$, and $\sigma^2 = 10^4$. In this case, the risk minimizer is $\theta^* = (\mu_1 + \mu_2)/2$.}
    \label{fig:exact_stopping}
\end{figure}
}\end{example}

The setting described in \Cref{ex:falsenull} is the ``best case" scenario in the sense that the practitioner gathers data until a null hypothesis is \textit{correctly} rejected; a meta-analysis of replication studies could plausibly correct this issue. But what happens when publications in the literature only present \textit{false} rejections of a null hypothesis?

\begin{example}\label{ex:truenull}\normalfont{
Consider the same setting as \Cref{ex:falsenull}, with a null hypothesis of the form $H_0: \theta^* \geq c$ for some $c$, but now suppose that $H_0$ is true. Due to the difficulty in publishing negative results, the only data sets present in the literature will be those that falsely reject this null hypothesis, and the coverage of these intervals is shown in \Cref{fig:reroll} to be essentially zero for most levels of nominal coverage. No meta-analysis can correct for this issue, as all published data are biased towards the incorrect alternative hypothesis. However, \Cref{fig:reroll} clearly demonstrates that the GUe confidence sets remain valid at all nominal coverage levels, even coming close to matching nominal coverage at all relevant levels.

\begin{figure}[t]
    \centering
    \includegraphics[width=0.45\textwidth]{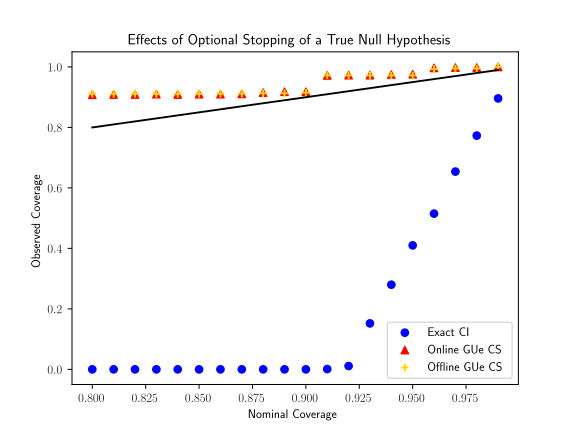}
    \caption{Nominal versus observed coverage of the ``exact" and GUe confidence sets when only considering data where the null hypothesis $H_0: \theta^* \geq -10$ is falsely rejected for $\frac{1}{2}\operatorname{N}(\mu_1, \sigma^2) + \frac{1}{2}\operatorname{N}(\mu_2, \sigma^2)$ data, with $\mu_1 = 5$, $\mu_2 = 10$, and $\sigma^2 = 10^4$. Note that $\theta^* = (\mu_1+\mu_2)/2$.}
    \label{fig:reroll}
\end{figure}
}\end{example}

\begin{example}\label{ex:cherrypicking}\normalfont{
Another common way for science to fail to be replicated is due to the unjustifiable removal of outliers.  Doing so can significantly reduce the standard errors and may appear to be justifiable---after all, one should surely remove data points that are corrupted by non-statistical errors.  To illustrate the effects of such cherry-picking, we simulate data that have ``outliers" removed using Tukey's fences criterion for outliers\footnote{In practice, the ``fences" are chosen in a data-driven manner. However, for the purposes of this simulation study, we use the true values based on the data-generating distribution so that the i.i.d. assumption holds.} with $k=1$. \Cref{fig:cherrypick} demonstrates the effects of unwarranted removal of outliers on the validity of confidence sets for data from the triangular distribution (with support $[0, 2]$ and a peak at $1$) and from $\operatorname{Beta}(5, 2)$. As usual, our proposed confidence sets (with learning rates chosen via nonparametric bootstrap) maintain at least the correct level of coverage, whereas the ``exact" confidence intervals fail to attain the nominal coverage level---and quite drastically so for the triangular-distributed data.
\begin{figure}[t]
    \centering
    \includegraphics[width=0.45\textwidth]{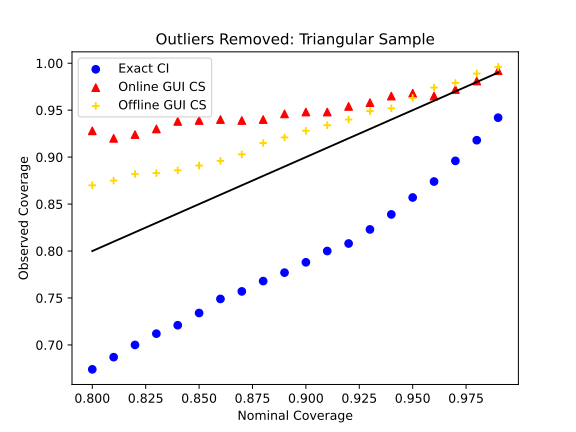}    \includegraphics[width=0.45\textwidth]{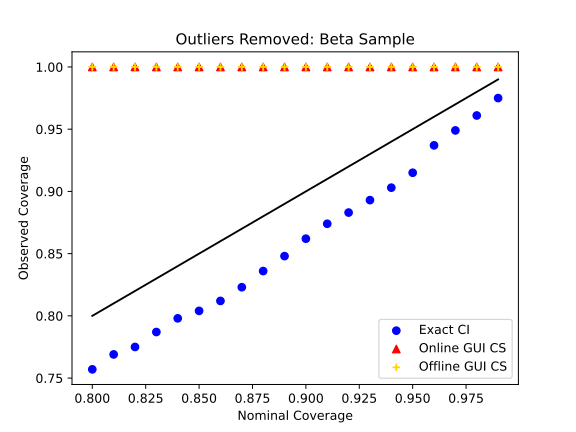}
    \caption{Coverage of the mean of $\operatorname{N}(0, 1)$ and $\operatorname{Beta}(5, 2)$ data when outliers are removed via the Tukey criterion ($k = 1$).}
    \label{fig:cherrypick}
\end{figure}
}\end{example}

Additional simulation studies can be found in \Cref{apdx:extrasims}:  One is a case where the true $\theta^*$ falls on the boundary of the parameter space; another is a case where the strong central condition fails; yet another is a situation where bootstrap is known to fail. In all cases, the GUe confidence intervals are shown empirically to be valid while existing methods fail to achieve the nominal coverage probability.

\section{Real data examples}\label{sec:realdata}
\subsection{Millikan's electron charge study, revisited}\label{sec:millikan}

The first experiment done to measure the charge on an electron was by \citet{millikan1913}.  About this experiment, \citet{feynman1974} noted the following: 
\begin{quote}
``Millikan measured the charge on an electron by an experiment with falling oil drops and got an answer which we now know not to be quite right... It's interesting to look at the history of measurements of the charge of the electron, after Millikan.  If you plot them as a function of time, you find that one is a little bigger than Millikan's, and the next one's a little bit bigger than that, and the next one's a little bit bigger than that, until finally they settle down to a number which is higher. Why didn't they discover that the new number was higher right away?  It's a thing that scientists are ashamed of—this history—because it's apparent that people did things like this: When they got a number that was too high above Millikan's, they thought something must be wrong—and they would look for and find a reason why something might be wrong.  When they got a number closer to Millikan's value they didn't look so hard.  And so they eliminated the numbers that were too far off...''
\end{quote}
Indeed, the charge of an electron is now known to be exactly $160.2176634\text{ zC}$, whereas Millikan's experiment yielded a point estimate of $159.2 \text{ zC}$ with standard error $0.07 \text{ zC}$. Thus, Millikan's point estimate was roughly 14 standard errors below the true value---in part due to Millikan's cherry-picking of data to artificially exclude data points he deemed to be outliers, using data-dependent versions of the ``fences" from \Cref{ex:cherrypicking}.

Follow-up papers that attempted to calculate the charge of an electron include \citet{wadlund1928} at $159.24 \text{ zC}$, \citet{backlin1929} at $159.88 \text{ zC}$, and \citet{bearden1931} at $160.31 \text{ zC}$. After the estimate of \citet{bearden1931}, the timeline of results reported by \citet{hill2021} suggests that later estimates all tended to fall quite close to the true value of about $160.2 \text{ zC}$.

To see how the GUe-value applies to Millikan's oil drop experiment, we use the nonparametric bootstrap to choose the learning rate for the GUe confidence set. We find that the offline GUe confidence set from Millikan's cherry-picked data consistently covers the true value of the charge of an electron until $\alpha \approx 0.065$, and continues sporadically covering the true value (up to fluctuations due to random sampling in the nonparametric bootstrap) until $\alpha \approx 0.22$. Had the uncertainty in measurement of the charge of the electron been calculated via the GUe-value, chemists might have converged to the correct value faster than the multiple decades it actually took, thanks to not being constrained by the too-narrow confidence interval generated by Millikan's cherry-picked data.

\subsection{Quantile regression of MyAnimeList ratings}
One challenging problem for traditional inference is quantile regression with non-i.i.d.~errors. That is, given a linear model $Y_i = X_i^\top\beta + \varepsilon_i$ where $\varepsilon_i$ are not identically distributed, we wish to estimate the conditional $q$-quantile $\theta^*_q$, which minimizes the risk corresponding to the loss function
\begin{equation*}
    \ell(\theta\altgiven x, y) = (y - x^\top\theta)\cdot \left\{q - \mathbbm{1}(y - x^\top\theta < 0)\right\}.
\end{equation*}
The theory of M-estimation yields that the asymptotic variance of the ERM depends on the conditional density of the response evaluated at the $q$ conditional quantile, which can be quite difficult to estimate---particularly for extreme quantiles (i.e., $q$ near $0$ or $1$). One commonly used approach is the Powell kernel density estimator \citep{powell1991}, though this is quite sensitive to the choice of kernel and bandwidth parameter. To illustrate this, \Cref{fig:maldata} presents 95\% prediction intervals for linear quantile regression with $q = 0.01$ using data from the website MyAnimeList (MAL), which provides user ratings (on a scale of 0 to 10) over time of animated media. As can be seen, the M-estimation based Powell standard errors greatly vary depending on the choice of bandwidth, and the standard choice used by the \texttt{rq} package in R \citep{koenker2024} certainly undercovers $\theta^*_q$, even for the reasonably large sample size of $n = 384$ given here. Consequently, it is evident that the ellipsoidal asymptotic confidence sets for $\theta^*_q$ are untrustworthy in this problem.

\begin{figure}[tbp]
    \centering
    \includegraphics[width=0.45\textwidth]{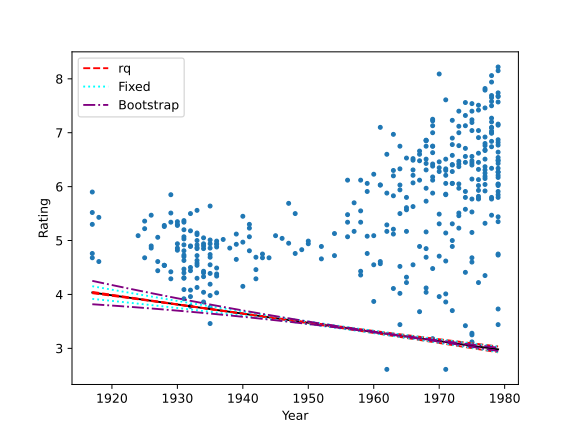}
    \caption{95\% prediction intervals for the conditional $0.01$-quantile of MAL scores before 1980. The standard error of the regression coefficient is estimated using the Powell method with two different choices of bandwidth (firstly as implemented in the \texttt{rq} package in R, and secondly using the fixed bandwidth 10 times that of \texttt{rq}), as well as nonparametric bootstrap.}
    \label{fig:maldata}
\end{figure}

\Cref{fig:contour} shows the contour plot for the 95\% offline GUe confidence set for $\theta^*_q$ after standardizing the covariate; for comparison, we also show the asymptotic 95\% ellipsoidal confidence sets suggested by M-estimation and bootstrap. Notably, the area of the GUe confidence sets is not inordinately larger than the bootstrap and Powell confidence sets. For the up-to-year-1980 data, the non-ellipsoidal shape of the GUe set reveals a directional component to the uncertainty in $\thetahat_q$ that the other large-sample methods are not able to suggest. For example, it indicates that the range of plausible values for the intercept is skewed, with far more values below the point estimate than above.

\begin{figure}[t]
    \centering
    \begin{subfigure}{0.45\textwidth}
    \includegraphics[width = \textwidth]{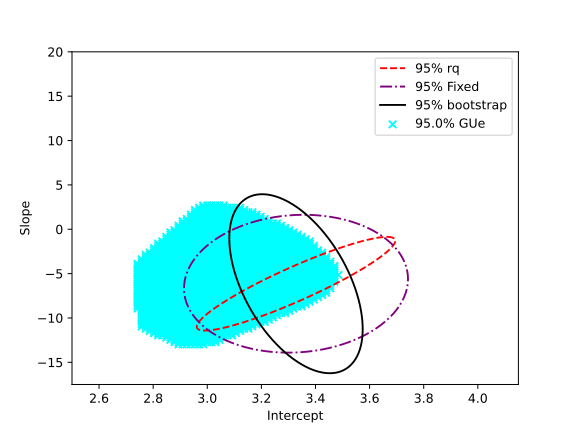}
    \end{subfigure}
    \begin{subfigure}{0.45\textwidth}
    \includegraphics[width = \textwidth]{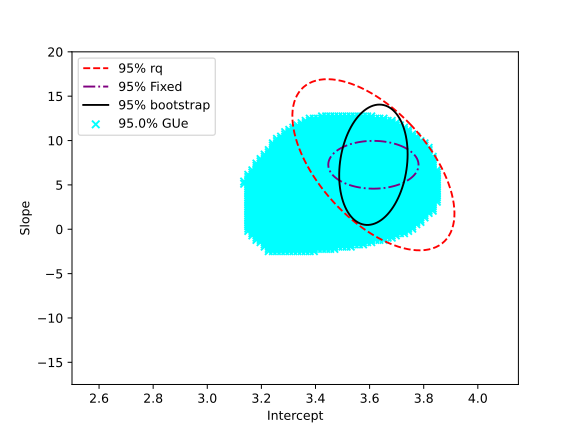}
    \end{subfigure}
    \caption{Contour plots of confidence sets for $\theta^*_q$. On the left are the confidence sets for MAL data before 1980 ($n = 384$); on the right are the confidence sets for data before 2000 ($n = 1999$)}
    \label{fig:contour}
\end{figure}

\section{Conclusion}\label{sec:conclusion}
In this paper we considered a context common in modern statistical learning problems concerned with risk minimization.  For such problems, we have proposed a new {\em generalized universal inference} framework that leverages the theory of e-values and e-processes, and have shown that the corresponding GUe-value tests and confidence sets for the unknown risk minimizer are provably valid in finite samples.  These validity conclusions do not come for free, as one might hope based on the developments in \citet{wasserman2020}, but they follow from a general and relatively mild condition called the strong central condition.  Furthermore, under certain weak consistency conditions, the diameter of the GUe-value confidence sets shrinks at the same rates achieved by the driving ERM, suggesting the finite-sample validity guarantees do not come at the cost of severe inefficiency.  The online GUe-value features an additional anytime-validity property that means the validity claims hold uniformly over all stopping rules used in the data collection process.  In particular, we showed that the method's reliable performance is stable across a variety of common stopping rules believed to contribute to the replication crisis in science.
Furthermore, the practitioner has agency in choosing how conservative they wish to be, as the methodology they use to choose the learning rate for the GUe-value can be influenced by the assumptions they are willing to make regarding the collected data.

The test consistency results in \Cref{thm:power,thm:radius} are related to a more fundamental question concerning the asymptotic growth rate of the proposed GUe-value, akin to the investigations in \citet{grunwald2020b} for the well-specified statistical model setting.  For a given (possibly composite) null hypothesis $H_0: \theta^* \in \Theta_0$, recall that $G_n^{(\omega)}(\Theta_0) = \inf_{\theta \in \Theta_0} G_n^{(\omega)}(\theta)$.  Following Theorem~2 in \citet{dixit2023}, our claim is that, under certain conditions (e.g., the fixed learning rate $\omega$ is sufficiently small), the asymptotic growth rate of our GUe-value is 
\[ \log G_n^{(\omega)}(\Theta_0) = n \times \omega \Bigl\{ \inf_{\theta \in \Theta_0} R(\theta) - \inf_{\theta \not\in \Theta_0} R(\theta) \Bigr\} + o(n), \quad \text{almost surely}. \]
Note that, if the hypothesis is true in the sense that $\Theta_0 \ni \theta^*$, then the GUe-value vanishes as $n \to \infty$, as expected.  Alternatively, if the hypothesis is false in the sense that $\Theta_0 \not\ni \theta^*$, then the GUe-value diverges to $\infty$ as $n \to \infty$, again as expected.  Moreover, the (exponential) rate at which these limits are approached corresponds, e.g., in the latter case, to the degree of separation between $\theta^*$ and $\Theta_0$ determined by the risk function: the further $\theta^*$ is from $\Theta_0$, as measured by $\omega \inf_{\theta \in \Theta_0} \{R(\theta) - R(\theta^*)\}$, the faster the growth rate.  Our further conjecture is that the asymptotic growth rate of the GUe-value above is the ``optimal growth rate for e-values aimed at inference on a risk minimizer,'' but we leave a proper formulation and verification of these claims for follow-up work.  

Future investigations will consider how this theory might extend to the non-i.i.d.~setting to allow for inference on risk minimizers in longitudinal or spatial data, for example. Another important open question is how best to choose the learning rate for the GUe-value (or other e-values that require data-driven tuning), and under what conditions the proposed bootstrapping strategy offers GUe-value confidence sets with provable validity guarantees. The theory we have presented for learning rate selection is quite limited, even for the special case of the $L^2$ loss function, despite how critical the choice of learning rate is in providing finite-sample validity guarantees for the GUe-value; thus, further work in this direction is necessary. Finally, we hope to further investigate the utility of the GUe-value in more modern machine learning models through its connection to the Gibbs posterior and thus PAC-Bayes learning theory.

\vskip 0.2in
\bibliography{references.bib}

\newpage

\appendix

\crefalias{section}{appendix}
\renewcommand{\thesubsubsection}{Part \Roman{subsubsection}}

\section{Technical remarks}
\label{app:remarks}

\begin{remark}\label{rem:sccsufficient}\normalfont{
As discussed in detail in \citet{vanErvan2015}, the strong central condition holds in a number of practically relevant cases; see, also, \citet{grunwald2020}.  First, if the learning problem is determined by a well-specified statistical model, as in \citet{wasserman2020} and many other papers, where the loss $\ell$ is the negative log-likelihood, then it follows from H\"older's inequality that the strong central condition holds with $\bar\omega=1$.  Even if the statistical model is incorrectly specified, under certain convexity conditions \citep[e.g.,][]{kleijn2006, deBlasi2013, ramamoorthi2015}, one can often demonstrate the strong central condition for some $\bar\omega < 1$; see, e.g., \citet{heide2020} for an application to misspecified generalized linear models.  Outside the context of a posited statistical model, the strong central condition holds for any bounded loss when the parameter space is convex.  This includes the typical classification problems based on zero-one loss, as well as variants that arise in, e.g., inference on the minimum clinically important difference \citep{xu.mcid, syring.martin.mcid}. For unbounded loss functions, such as the $L^p$ losses, further restrictions on the data-generating process are required in order for the strong central condition to hold.  For instance, Example~11 of \citet{grunwald2020} notes that the strong central condition cannot hold for the $L^2$ loss without subexponential tail decay on the data-generating process, or some other combination of convexity, boundedness, etc.; see, also, \citet[][Example 4.20]{vanErvan2015}. 
Nevertheless, subexponential tail assumptions are common in the literature, so there are many practical applications in which the strong central condition can be verified for $L^p$ losses. 
}\end{remark}

\begin{remark}\label{rem:infnotexist}\normalfont{
Note that although confidence sets only make sense when the risk minimizer $\theta^*$ exists, \Cref{thm:power,thm:radius} 
apply even if $\inf_\vartheta R(\vartheta)$ is never attained.  Two instances where the infimum risk fails to be attained include models where the parameter space is not compact (such as when $\theta$ represents a variance component that lies in $(0, \infty)$) and those that use risk functions that are non-coercive (such as the cross-entropy loss). Indeed, the most common example where the infimum fails to be attained is when $\theta$ denotes the parameter in logistic regression and the population is separated---i.e., for the population $P$, where $P \subseteq \R^p\times\{0, 1\}$, there exists $\beta\in\R^p$ such that for any $(x, y)\in P$, we have that $\beta^\top x > 0$ implies $y = 1$ and $\beta^\top x < 0$ implies $y = 0$---as separation forces at least one component of $\theta^*$ to be infinite \citep[][]{albert1984}.  
Even in such cases, the theorems guarantee that the GUe-value grows large on all of $\Theta$.  Consequently, the corresponding confidence sets shrink to the empty set as more data are collected. This may indicate to the user that their statistical learning problem is ill-posed, if they were not aware of this already.
}\end{remark}

\begin{remark}\label{rem:safebayes}\normalfont{
As discussed in the main text, in addition to the GPC algorithm of \citet{syring2019} that we have adopted here, there are a number of strategies available for choosing the learning rate in the construction of a Gibbs or generalized Bayes posterior distribution, including those found in \citet{bissiri2016}, \citet{holmes2017}, \citet{lyddon2019}, \citet{grunwald2017}, \citet{perrotta2020}.  It was our initial conjecture that the GUe-value and its properties would not be particularly sensitive to the choice of algorithm used to choose the learning rate.  More specifically, we expected that the sequential aspect of Gr\"unwald's SafeBayes strategy made it particularly well-suited to this application, perhaps even better suited here than in the non-sequential applications of generalized Bayes.  So, to our surprise, there was a difference in the GUe-value across different learning rate strategies. And even more surprising was that, despite seeming well-suited for this application, the GUe-value with learning rate chosen by SafeBayes failed to maintain validity. \Cref{fig:rsafevsgpc} illustrates the coverage of the GUe-value when using learning rates chosen by SafeBayes and by GPC calibrated to $\alpha = 0.10$; the data was generated from the logistic model $(Y \mid X) \sim \operatorname{Bernoulli}(\operatorname{expit}(1\cdot X - 1))$ with $X \sim N(0, 1)$, and the loss function used was the Savage loss. It is clear that while the GPC-based approach is clearly approximately calibrated, the coverage of SafeBayes deteriorates as the sample size grows. It is an interesting open question why GPC seems to work particularly well here and if there are other strategies that are even better suited to GUe-values than GPC.  
\begin{figure}[t]
    \centering
    \includegraphics[width=0.5\linewidth]{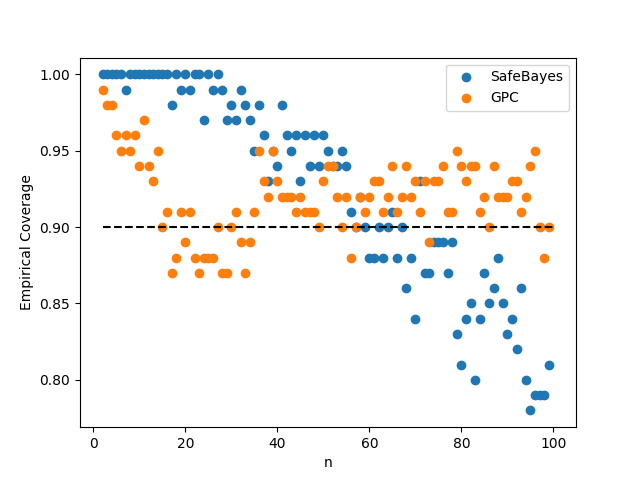}
    \caption{Observed coverage of GUe confidence sets when using SafeBayes and GPC to select learning rates on on i.i.d. logistic data of size $n$, using the Savage loss and $\alpha = 0.10$.}
    \label{fig:rsafevsgpc}
\end{figure}
}\end{remark}

\section{Proofs}\label{sec:appendix}
\subsection{Proof of \Cref{lem:eprocess}}
This essentially follows from Lemma 3 of \citet{wang2023}, but we give the proof here for completeness. We first show that $E_n := G_{n, \text{on}}(\theta^*)$ is a non-negative supermartingale. For convenience of notation, define $\Delta_i := \ell(\thetahat_{i-1}\altgiven Z_i) - \ell(\theta^*\altgiven Z_i)$, for $i=1,2,\ldots$, where, again, $\thetahat_0$ is a fixed constant.  Then 
\begin{align*}
\E(E_n \mid Z^{n-1}) & = \E\left\{\exp(-\sum_{i=1}^n \omega\Delta_i) \bigmid Z^{n-1}\right\} \\
& = \E\left\{\exp(-\sum_{i=1}^{n-1} \omega\Delta_i) \cdot \exp(-\omega \Delta_n) \bigmid Z^{n-1}\right\} \\
& = E_{n-1} \cdot  \E\left\{\exp(-\omega \Delta_n) \mid Z^{n-1}\right\},
\end{align*}
where the last equality follows because $\sum_{i=1}^{n-1} \Delta_i$ is a measurable function of $Z^{n-1}$.  Since $Z_n$ and $Z^{n-1}$ are independent and $\thetahat_{n-1}$ is a measurable function of $Z^{n-1}$, the latter conditional expectation in the above display can be re-expressed as 
\[ \E \exp\bigl[-\omega\{ \ell(\vartheta; Z) - \ell(\theta^*; Z) \} \bigr], \quad \text{for some fixed $\vartheta \in \Theta$ and $\omega\in[0, \bar{\omega})$}, \]
and is bounded by 1, by the strong central condition; thus, $E_n = G_{n,\text{on}}(\theta^*)$ is a non-negative supermartingale.  Finally, since  
\[ \E(E_1) = \E \exp[-\omega \{ \ell(\thetahat_0; Z_1) - \ell(\theta^*; Z_1) \}] \leq 1, \]
again by the strong central condition, it follows by a variant of the optional stopping theorem \citep[e.g.,][Theorem 4.8.4]{durrett2019} that $E_n = G_{n,\text{on}}(\theta^*)$ is an e-process.

\subsection{Proof of \Cref{lem:evalue}}
Again, for convenience, define $\Delta_i := \ell(\thetahat_{S_1} \altgiven Z_i) - \ell(\theta^*\altgiven Z_i)$, for $i=1,2,\ldots, n_2$ where each $Z_i \in S_2$. Since $\thetahat_{S_1}$ is a measurable function of $S_1$, the strong central condition implies that $\E\{\exp(-\omega \Delta_i) \mid S_1\} \leq 1$ for each $i=1,2,\ldots$. We hence have that
\begin{equation*}
    \E\{G_{S,\text{off}}(\theta^*) \mid S_1\} = \E\left\{\exp(-\omega \sum_{i=1}^{n_2} \Delta_i) \bigmid S_1\right\} 
    = \prod_{i=1}^{n_2} \E\left\{\exp(-\omega\Delta_i) \bigmid S_1\right\} \leq 1
\end{equation*}
since the $\Delta_i$ are independent given $S_1$.  The law of iterated expectations gives 
\begin{equation*}
    \E\{G_{S,\text{off}}(\theta^*)\} = \E\E\{G_{S,\text{off}}(\theta^*) \mid S_1\} \leq 1,
\end{equation*}
and so the offline GUe-value is indeed an e-value.

\subsection{Proof of \Cref{thm:evalue}}
Since $\Theta_0$ contains $\theta^*$, it follows that $G_n(\Theta_0) \leq G_n(\theta^*)$.  Then Markov's inequality and \Cref{lem:evalue} gives
\[ \Pr\{G_n(\Theta_0) \geq \alpha^{-1}\} \leq \Pr\{ G_n(\theta^*) \geq \alpha^{-1} \} \leq \alpha \E\{ G_n(\theta^*) \} \leq \alpha, \]
which proves the first claim.  The coverage probability claim follows since $C_\alpha(Z^n) \not\ni \theta^*$ if and only if $G_n(\theta^*) \geq \alpha^{-1}$, and the latter event has probability at most $\alpha$ as just shown.  The final two claims follow from the same arguments given above, thanks to the fact that $G_n(\theta^*)$ is an e-process, as shown in \Cref{lem:eprocess}. 

\subsection{Proof of \Cref{thm:power}}\label{prf:thmpower}

The following lemma is of use in the proofs of \Cref{thm:power,thm:radius}. 
\begin{lemma}\label{lem:onlinedomoffline}
Let $\thetahat_n$ be an $(\varepsilon, \delta)$-AERM for all $n$. Further, suppose that there exists $\overline{\Omega}$ such that $0 \leq \omegahat_n \leq \overline{\Omega}$ for all $n$. Then for any $n$, we have that
\begin{equation*}
    \sum_{i=1}^n \omegahat_{i-1}\ell(\thetahat_i\altgiven Z_i) \leq \overline{\Omega}\left\{H_n^{(\varepsilon)}\delta + \sum_{i=1}^n \ell(\thetahat_n\altgiven Z_i)\right\}
\end{equation*}
where $H_n^{(m)} := \sum_{k=1}^n k^{-m}$ denotes the generalized harmonic number of order $m$.
\end{lemma}
\begin{proof}
We proceed by induction. When $n = 1$, the statement trivially holds. Thus, suppose 
\begin{equation*}
    \sum_{i=1}^{n-1} \omegahat_{i-1}\ell(\thetahat_i\altgiven Z_i) \leq \overline{\Omega}\left\{H_{n-1}^{(\varepsilon)}\overline{\Omega}\delta + \sum_{i=1}^{n-1} \ell(\thetahat_{n-1}\altgiven Z_i)\right\}.
\end{equation*}
Then 
\begin{align*}
    \sum_{i=1}^n \omegahat_{i-1}\ell(\thetahat_i\altgiven Z_i)
    &= \left\{\sum_{i=1}^{n-1} \omegahat_{i-1}\ell(\thetahat_i\altgiven Z_i)\right\} + \omegahat_{n-1}\ell(\thetahat_n\altgiven Z_n) \\
    &\leq \overline{\Omega}\left\{H_{n-1}^{(\varepsilon)}\delta + \sum_{i=1}^{n-1}\ell(\thetahat_{n-1}\altgiven Z_i)\right\} + \overline{\Omega}\ell(\thetahat_n\altgiven Z_n) \\
    &\leq \overline{\Omega}\qty[H_{n-1}^{(\varepsilon)}\delta + n^{-\varepsilon}\delta + \left\{\inf_{\vartheta\in\Theta} \sum_{i=1}^{n-1}\ell(\vartheta \altgiven Z_i)\right\} + \ell(\thetahat_n\altgiven Z_n)] \\
    &\leq \overline{\Omega}\qty[H_n^{(\varepsilon)}\delta + \left\{\sum_{i=1}^{n-1}\ell(\thetahat_{n}\altgiven Z_i)\right\} + \ell(\thetahat_n\altgiven Z_n)] \\
    &= \overline{\Omega}\left\{H_n^{(\varepsilon)}\delta +  \sum_{i=1}^n \ell(\thetahat_n\altgiven Z_i)\right\}
\end{align*}
where the first inequality uses the inductive hypothesis, the second inequality uses the definition of the $(\varepsilon, \delta)$-AERM $\thetahat_{n-1}$, and the third inequality uses the definition of the infimum.
\end{proof}

The proof of \Cref{thm:power} proceeds in two parts---\Cref{part:thm2online} proves the consistency result for the online GUe-value, and \Cref{part:thm2offline} does the same for the offline GUe-value.

\subsubsection{Online GUe}\label{part:thm2online}
For convenience in notation, define $\sigma_n := \frac{1}{n}\sum_{i=1}^n\omegahat_{i-1}\left\{\ell(\thetahat_{i-1}\altgiven Z_i) - \ell(\thetahat_{n}\altgiven Z_i)\right\}$. Further, let $\theta$ be such that $R(\theta) > \inf_{\vartheta}R(\vartheta)$, and define $\Delta := \omegaunder\cdot(R(\theta) - \inf_{\vartheta}R(\vartheta))/3$. We then have that
\begin{align*}
\Pr\left\{\widehat{G}_{n, \text{on}}(\theta) \geq \frac{1}{\alpha} \right\}
    &= \Pr\left\{\sum_{i=1}^n \omegahat_{i-1}\left\{\ell(\thetahat_{i-1}\altgiven Z_i) - \ell(\theta\altgiven Z_i)\right\} \leq \log(\alpha) \right\} \\
    &= \Pr\qty[\sum_{i=1}^n \omegahat_{i-1}\left\{\ell(\thetahat_{i-1}\altgiven Z_i) - \ell(\thetahat_n\altgiven Z_i) + \ell(\thetahat_n\altgiven Z_i) - \ell(\theta\altgiven Z_i)\right\} \leq \log \alpha] \\
    &= \Pr\qty[\sigma_n + \frac{1}{n}\sum_{i=1}^n \omegahat_{i-1}\left\{\ell(\thetahat_n\altgiven Z_i) - \ell(\theta\altgiven Z_i)\right\} \leq \frac{\log\alpha}{n}] \\
    &\geq \underbrace{\Pr\qty(\sigma_n\leq \frac{\Delta}{2} )}_{\text{(A)}} + \underbrace{\Pr\qty[\frac{1}{n}\sum_{i=1}^n \omegahat_{i-1}\left\{\ell(\thetahat_n\altgiven Z_i) - \ell(\theta\altgiven Z_i)\right\} \leq \frac{\log\alpha}{n}  -\frac{\Delta}{2}]}_{\text{(B)}} - 1.
\end{align*}

% \begin{align*}
%     &\Pr\left\{G_{n, \text{on}}^{(\omega)}(\theta) \geq 1/\alpha\right\} \\
%     &= \Pr\left\{\sum_{i=1}^n \ell(\thetahat_{i-1}\altgiven Z_i) - \ell(\theta^*\altgiven Z_i) \leq \log(\alpha)/\omega\right\} \\
%     &= \Pr\left\{\sum_{i=1}^n \ell(\thetahat_{i-1}\altgiven Z_i) - \ell(\thetahat_n\altgiven Z_i) + \ell(\thetahat_n\altgiven Z_i) - \ell(\theta^*\altgiven Z_i) \leq \log(\alpha)/\omega\right\} \\
%     &= \Pr\left\{\sum_{i=1}^n \qty(\ell(\thetahat_{i-1}\altgiven Z_i) - \ell(\thetahat_n\altgiven Z_i)) + n\rhat_n(\thetahat_n) - n\rhat_n(\theta^*) \leq \log(\alpha)/\omega\right\} \\
%     &= \Pr\left\{\frac{1}{n}\sum_{i=1}^n \qty(\ell(\thetahat_{i-1}\altgiven Z_i) - \ell(\thetahat_n\altgiven Z_i)) + \rhat_n(\thetahat_n) - \rhat_n(\theta^*) \leq \log(\alpha)/(n\omega)\right\} \\
%     &\geq \Pr\left\{\frac{1}{n}\sum_{i=1}^n \ell(\thetahat_{i-1}\altgiven Z_i) - \ell(\thetahat_n\altgiven Z_i) \leq \Delta/2 \right\} \\
%         &\qquad\qquad + \Pr\left\{\rhat_n(\thetahat_n) - \rhat_n(\theta^*) \leq \log(\alpha)/(n\omega) - \Delta/2 \right\} - 1
% \end{align*}
% where we choose $\Delta = (R(\theta) - \inf_\vartheta R(\vartheta))/2$, which is positive since we assumed that $R(\theta) > \inf_\vartheta R(\vartheta)$. 
We need to show that both terms (A) and (B) go to $1$ as $n\rightarrow \infty$. For the former, we have by \Cref{lem:onlinedomoffline} that
\begin{equation*}
    \Pr\qty(\sigma_n \leq \frac{\Delta}{2})  \geq \Pr\left[\frac{1}{n}\sum_{i=1}^n \omegahat_{i-1}\left\{\ell(\thetahat_{i-1}\altgiven Z_i) - \ell(\thetahat_i\altgiven Z_i)\right\} \leq \frac{\Delta}{2} -\overline{\Omega}\cdot \frac{H_n^{(\varepsilon)}}{n}\delta\right].
\end{equation*}
% \begin{equation*}
%     \Pr\left\{\frac{1}{n}\sum_{i=1}^n \ell(\thetahat_{i-1}\altgiven Z_i) - \ell(\thetahat_n\altgiven Z_i) \leq \Delta/2 \right\}  \geq \Pr\left\{\frac{1}{n}\sum_{i=1}^n \ell(\thetahat_{i-1}\altgiven Z_i) - \ell(\thetahat_i\altgiven Z_i) \leq \Delta/2 \right\} 
% \end{equation*}
Next, we have by the stability hypothesis that each $\ell(\thetahat_{i-1}\altgiven Z_i) - \ell(\thetahat_{i}\altgiven Z_i) \leq \beta_i$ for some positive scalars $\beta_i$ that satisfy $\lim\limits_{n\rightarrow \infty} \beta_n = 0$. This implies that $\frac{1}{n}\sum \overline{\Omega} \beta_i \rightarrow 0$, so
\begin{equation*}
    \lim_{n\rightarrow\infty} \Pr\left\{\sigma_n \leq \frac{\Delta}{2} \right\}  \geq \lim_{n\rightarrow\infty} \Pr\left\{\frac{1}{n}\sum_{i=1}^n \overline{\Omega} \beta_i \leq \frac{\Delta}{2}-\overline{\Omega}\cdot \frac{H_n^{(\varepsilon)}}{n}\delta\right\} = 1
\end{equation*}
% \begin{equation*}
%     \lim_{n\rightarrow\infty}\Pr\left\{\frac{1}{n}\sum_{i=1}^n \ell(\thetahat_{i-1}\altgiven Z_i) - \ell(\thetahat_n\altgiven Z_i) \leq \Delta/2 \right\}  \geq \lim_{n\rightarrow\infty} \Pr\left\{\frac{1}{n}\sum_{i=1}^n \beta_i \leq \Delta/2\right\} = 1
% \end{equation*}
as desired, since each $\omegahat_{i-1} \leq \overline{\Omega}$ and for large enough $n$ we have that $\Delta/2 - \overline{\Omega}H_n^{(\varepsilon)}\cdot n^{-1}\delta > 0$ by our hypothesis concerning $\delta$.

To show that (B) has limit $1$, let $\omega := \limsup_n \omegahat_n$, and note that
\begin{align*}
&\Pr\qty[\frac{1}{n}\sum_{i=1}^n \omegahat_{i-1}\left\{\ell(\thetahat_n\altgiven Z_i) - \ell(\theta\altgiven Z_i)\right\} \leq \frac{\log\alpha}{n}  -\frac{\Delta}{2}] \\
&=\Pr[\frac{1}{n}\sum_{i=1}^n \omegahat_{i-1}\ell(\thetahat_n\altgiven Z_i) - \omega  R(\thetahat_n) + \omega R(\thetahat_n) - \omegahat_{i-1}\ell(\theta\altgiven Z_i) \leq \frac{\log \alpha}{n} - \frac{\Delta}{2}] \\
&\geq \underbrace{\Pr\left\{\frac{1}{n}\sum_{i=1}^n \omegahat_{i-1}\ell(\thetahat_n\altgiven Z_i) - \omega  R(\thetahat_n) \leq \frac{\Delta}{2}\right\}}_{\text{(C)}} + \underbrace{\Pr\left\{\omega R(\thetahat_n) - \frac{1}{n}\sum_{i=1}^n\omegahat_{i-1}\ell(\theta\altgiven Z_i)\leq \frac{\log\alpha}{n} - \Delta\right\}}_{\text{(D)}} - 1.
\end{align*}
It now suffices to show that each of (C) and (D) also has limit $1$.
To do so for (C), we rewrite
\begin{equation*}
    \frac{1}{n}\sum_{i=1}^n \omegahat_{i-1}\ell(\thetahat_n\altgiven Z_i) - \omega  R(\thetahat_n) = \frac{1}{n}\sum_{i=1}^n (\omegahat_{i-1} - \omega)\ell(\thetahat_n\altgiven Z_i) + \omega\{\rhat_n(\thetahat_n) - R(\thetahat_n)\}.
\end{equation*}
The second term converges to zero in probability as $n\rightarrow\infty$ by the uniform converges of $\rhat_n$ to $R$. For the first term, note that for any $\varepsilon > 0$ and any particular sample $(z_1, z_2, \ldots) \in \mathcal{Z}^\infty$, there exists $N$ such that for all $i > N$, $\omegahat_{i-1} - \omega < \varepsilon/(4\inf_\vartheta R(\vartheta))$. Hence, splitting $\sum_{i=1}^n a_i = \sum_{i=1}^N a_i + \sum_{i=N+1}^n a_i$, we arrive at
\begin{align*}
&\Pr\left\{\frac{1}{n}\sum_{i=1}^n (\omegahat_{i-1} - \omega)\ell(\thetahat_n\altgiven Z_i) < \varepsilon \right\} \geq \Pr\left\{\frac{\overline{\Omega} N}{n}\rhat_N(\thetahat_n) < \frac{\varepsilon}{2}\right\} + \Pr\left\{\frac{\varepsilon}{4 \inf_\vartheta R(\vartheta)}\rhat_n(\thetahat_n) < \frac{\varepsilon}{2}\right\} - 1.
\end{align*}
The second addend has limit 1, since uniform convergence of $\rhat_n$ to $R$ yields that $\rhat_n(\thetahat_n) \ip \inf_\vartheta R(\vartheta)$. For the first addend, we note by the stability hypothesis that $|\rhat_N(\thetahat_n) - \rhat_N(\thetahat_N)| < N\sum_{i=N+1}^n \beta_i$ for some positive scalars $\beta_i$ converging to zero. Hence,
\begin{equation*}
    \Pr\left\{\frac{\overline{\Omega} N}{n}\rhat_N(\thetahat_n) < \frac{\varepsilon}{2}\right\} \geq \Pr\left\{\frac{\overline{\Omega} N}{n}\rhat_N(\thetahat_N) < \frac{\varepsilon}{4}\right\} + \Pr\left\{\overline{\Omega} N^2 \cdot \frac{1}{n}\sum_{i=1}^n \beta_i < \frac{\varepsilon}{4}\right\} - 1.
\end{equation*}
Both of these addends now clearly have limit 1 as $n\rightarrow\infty$, so we are done with (C).

To show that (D) has limit $1$, we see that
% \begin{align*}
% &\Pr\left\{R(\thetahat_n) - \rhat_n(\theta) \leq \frac{\log\alpha}{n\omega} - \Delta\right\} \\
% &=\Pr\left\{R(\thetahat_n) - \inf_\vartheta R(\vartheta) + \inf_\vartheta R(\vartheta) - \rhat_n(\theta) \leq \frac{\log\alpha}{n\omega} - \Delta\right\} \\
% &\geq \Pr\left\{R(\thetahat_n) - \inf_\vartheta R(\vartheta) \leq \frac{\log \alpha}{n\omega} + \Delta\right\}.
% \end{align*}
% Now if $n \geq 2\log(1/\alpha)/(\omega\Delta)$, the above is lower bounded by
% \begin{equation*}
%     \Pr\left\{R(\thetahat_n) - \inf_\vartheta R(\vartheta) \leq \frac{\Delta}{2}\right\},
% \end{equation*}
% which we again know to have limit $1$ as $n\rightarrow\infty$ by uniform convergence.
\begin{align*}
    &\omega R(\thetahat_n) - \frac{1}{n}\sum_{i=1}^n \omegahat_{i-1}\ell(\theta\altgiven Z_i) \\
    &=\omega\left\{R(\thetahat_n) - \inf_{\vartheta}R(\vartheta)\right\} + \omega\left\{\inf_{\vartheta}R(\vartheta) - R(\theta)\right\} + \omega R(\theta) - \frac{1}{n}\sum_{i=1}^n \omegahat_{i-1}\ell(\theta\altgiven Z_i)
\end{align*}
Hence, bounding $\omega$ from below by $\omegaunder$, (D) is lower bounded by
\begin{equation*}
    \Pr\left\{R(\thetahat_n) - \inf_{\vartheta}R(\vartheta) \leq \overline{\Omega}^{-1}\Delta\right\} + \Pr\left\{\omega R(\theta) - \frac{1}{n}\sum_{i=1}^n \omegahat_{i-1}\ell(\theta\altgiven Z_i)\leq \frac{\log\alpha}{n} + \Delta \right\} - 1
\end{equation*}
The former term converges to one by uniform convergence of $\rhat_n$ to $R$; for the latter term, note that the right-hand side of the inequality is positive for all $n > \log(1/\alpha)/\Delta$, and so converges to one by a similar argument to (C).

\subsubsection{Offline Gue}\label{part:thm2offline}

For convenience in notation, define the function $\Phi(\vartheta) := \rhat_{S_2}(\vartheta) - R(\vartheta)$. Further, let $\theta$ be such that $R(\theta) > \inf_{\vartheta}R(\vartheta)$, and define $\Delta := (R(\theta) - \inf_{\vartheta}R(\vartheta))/3$. We then have that
\begin{align*}
    &\Pr\left\{\widehat{G}_{S, \text{off}}(\theta) \geq \frac{1}{\alpha}\right\} \\
    &= \Pr\left\{\rhat_{S_2}(\thetahat_{S_1}) - \rhat_{S_2}(\theta) \leq \frac{\log \alpha}{\omegahat_{S_1} n_2}\right\} \\
    &= \Pr\left\{[\rhat_{S_2}(\thetahat_{S_1}) - R(\thetahat_{S_1})] + [R(\thetahat_{S_1}) - R(\theta)] + [R(\theta)- \rhat_{S_2}(\theta)] \leq \frac{\log \alpha}{\omegahat_{S_1} n_2}\right\} \\
    &\geq \underbrace{\Pr\left\{\Phi(\thetahat_{S_1})\leq \Delta\right\}\vphantom{\left\{\frac{\log\alpha}{\omegahat_{S_1}}\right\}}}_{\text{(A)}} + \underbrace{\Pr\left\{R(\thetahat_{S_1}) - R(\theta) \leq \frac{\log \alpha}{\omegahat_{S_1} n_2} - 2\Delta\right\}}_{\text{(B)}} + \underbrace{\Pr\left\{-\Phi(\theta) \leq \Delta\right\}\vphantom{\left\{\frac{\log\alpha}{\omegahat_{S_1}}\right\}}}_{\text{(C)}} - 2.
\end{align*}
% \begin{align*}
%     &\Pr\left\{G_{S, \text{off}}^{(\omega)}(\theta) \geq 1/\alpha\right\} \\
%     &= \Pr\left\{\rhat_{S_2}(\thetahat_{S_1}) - \rhat_{S_2}(\theta) \leq \frac{\log \alpha}{\omega n_2}\right\} \\
%     &= \Pr\left\{[\rhat_{S_2}(\thetahat_{S_1}) - R(\thetahat_{S_1})] + [R(\thetahat_{S_1}) - R(\theta)] + [R(\theta)- \rhat_{S_2}(\theta)] \leq \frac{\log \alpha}{\omega n_2}\right\} \\
%     &\geq \Pr\left\{\rhat_{S_2}(\thetahat_{S_1}) - R(\thetahat_{S_1}) \leq \Delta \text{ and } R(\thetahat_{S_1}) - R(\theta) \leq \frac{\log \alpha}{\omega n_2} - 2\Delta \text{ and } R(\theta)- \rhat_{S_2}(\theta) \leq \Delta\right\} \\
%     &\geq \Pr\left\{\rhat_{S_2}(\thetahat_{S_1}) - R(\thetahat_{S_1}) \leq \Delta\right\} + \Pr\left\{R(\thetahat_{S_1}) - R(\theta) \leq \frac{\log \alpha}{\omega n_2} - 2\Delta\right\}
%     \\&\qquad\qquad+ \Pr\left\{R(\theta)- \rhat_{S_2}(\theta) \leq \Delta\right\} - 2
% \end{align*}
% where we specifically choose $\Delta = (R(\theta) - \inf_\vartheta R(\vartheta))/3$, which is positive since we assumed $R(\theta) > \inf_\vartheta R(\vartheta)$. We show that each of the addends has limit $1$.

It suffices to show that each of (A) (B), and (C) has limit $1$ as $(n_1, n_2) \rightarrow (\infty, \infty)$. For (A), we have from uniform convergence in probability of $\rhat_S$ to $R$ that for any $\varepsilon > 0$, there exists an $N\in\N$ such that if $n_2\geq N$,
% \begin{equation*}
%     1 - \Pr\left\{\sup_{\vartheta\in\Theta}|\rhat_{S_2}(\vartheta) - R(\vartheta)| \leq \Delta\right\} < \varepsilon.
% \end{equation*}
\begin{equation*}
    1 - \Pr\left\{\sup_{\vartheta\in\Theta}|\Phi(\vartheta)| \leq \Delta\right\} < \varepsilon.
\end{equation*}
We then note that for any $n_1\in\N$, if $\sup_{\vartheta\in\Theta}|\Phi(\vartheta)| \leq \Delta$, it is certainly also the case that $\Phi(\thetahat_{S_1}) \leq \Delta$. Hence, we have for any $\varepsilon > 0$ that there exists an $N\in\N$ such that for any $n_1\in\N$, if $n_2 \geq N$,
\begin{equation*}
    1 - \Pr\left\{\Phi(\thetahat_{S_1}) \leq \Delta\right\} < \varepsilon.
\end{equation*}
That is to say that as $n_2\rightarrow\infty$, $\Pr\left\{\Phi(\thetahat_{S_1}) \leq \Delta\right\}\rightarrow 1$ uniformly in $n_1$. Since this uniform limit does not depend on the value of $n_1$, we have that the double limit  exists and is equal to the single limit:
\begin{equation*}
    \lim_{(n_1, n_2) \rightarrow (\infty, \infty)}\Pr\left\{\Phi(\thetahat_{S_1}) \leq \Delta\right\} = \lim_{n_2 \rightarrow \infty}\Pr\left\{\Phi(\thetahat_{S_1}) \leq \Delta\right\} = 1.
\end{equation*}

We now examine term (B):
\begin{align*}
    &\Pr\left\{R(\thetahat_{S_1}) - R(\theta) \leq \frac{\log \alpha}{\omegahat_{S_1} n_2} - 2\Delta\right\} \\
    &= \Pr\left\{R(\thetahat_{S_1}) - \inf_\vartheta R(\vartheta) \leq \frac{\log \alpha}{\omegahat_{S_1} n_2} - 2\Delta + R(\theta) - \inf_\vartheta R(\vartheta)\right\} \\
    &= \Pr\left\{R(\thetahat_{S_1}) - \inf_\vartheta R(\vartheta) \leq \frac{\log \alpha}{\omegahat_{S_1} n_2} + \Delta\right\}
\end{align*}
where the final equality comes from our choice for $\Delta$. We now show that the double limit of the above expression exists and is equal to 1. To this end, let $\varepsilon > 0$ be arbitrary. Since $R(\thetahat_{S_1}) \ip \inf_\vartheta R(\vartheta)$, we have that there exists $N\in\Z^+$ such that if $n_1 \geq N$,
\begin{equation*}
    \Pr\left\{|R(\thetahat_{S_1}) - \inf_\vartheta R(\vartheta)| \leq \frac{\Delta}{2}\right\} > 1-\varepsilon.
\end{equation*}
Similarly, there exists $M\in\mathbb{Z}^+$ such that if $n_1 \geq M$, $\omegahat_{S_1} > \omegaunder/2$. Thus, if $n_1, n_2 \geq \max(-\frac{4\log \alpha}{\omegaunder \Delta}, N, M)$, we have that
\begin{align*}
    &\Pr\left\{R(\thetahat_{S_1}) - \inf_\vartheta R(\vartheta) \leq \frac{\log \alpha}{\omegahat_{S_1} n_2} + \Delta\right\}
    \\&\geq \Pr\left\{R(\thetahat_{S_1}) - \inf_\vartheta R(\vartheta) \leq -\frac{\log \alpha}{4\omegahat_{S_1} \cdot \log(\alpha)/(\omegaunder \Delta)} + \Delta \right\} \\
    &\geq \Pr\left\{R(\thetahat_{S_1}) - \inf_\vartheta R(\vartheta) \leq \frac{\Delta}{2}\right\} \\
    &\geq \Pr\left\{|R(\thetahat_{S_1}) - \inf_\vartheta R(\vartheta)| \leq \frac{\Delta}{2}\right\} \\
    &> 1- \varepsilon
\end{align*}
and thus our double limit is $1$:
\begin{equation*}
    \lim_{(n_1, n_2) \rightarrow(\infty, \infty)} \Pr\left\{R(\thetahat_{S_1}) - R(\theta) \leq \frac{\log \alpha}{\omega n_2} - 2\Delta\right\} = 1.
\end{equation*}

Finally, we note that (C) has limit $1$ by the law of large numbers, as $\E[\Phi(\theta)] = 0$.

\subsection{Proof of \Cref{thm:radius}}
We again divide the proof of the theorem in three parts, with each part proving the corresponding numbered result in \Cref{thm:radius}.

\subsubsection{Online GUe}
Define $\sigma_n := \frac{1}{n}\sum_{i=1}^n \omegahat_{i-1}\left\{\ell(\thetahat_{i-1}\altgiven Z_i) - \ell(\thetahat_n\altgiven Z_i)\right\}$, and $\Phi_n(\theta) := \frac{1}{n}\sum_{i=1}^n \omegahat_{i-1}\ell(\theta\altgiven Z_i) - \omega R(\thetahat_n)$, where $\omega = \limsup_n \omegahat_n$. Further define $\Delta_n := \omegaunder \cdot (R(\theta_n) - \inf_{\vartheta}R(\vartheta))/3$; note that there exists $c > 0$ such that $\Delta_n \geq c\cdot n^{-\beta}$ due to the definition of $(\theta_n)_{n\in\N}$. Then using the same arguments as in \Cref{part:thm2online} of the proof of Theorem 3, we have that
\begin{align*}
    &\Pr\left\{\widehat{G}_{n, \text{on}}(\theta_n) \geq \frac{1}{\alpha}\right\} \\
    &\geq \Pr(\sigma_n\leq \frac{\Delta_n}{2}) + \Pr\left\{\Phi_n(\thetahat_n) \leq \frac{\Delta_n}{2}\right\} +\Pr\left\{-\Phi_n(\theta_n) \leq \frac{\log\alpha}{n} - \Delta_n\right\} - 2.
\end{align*}
% \begin{align*}
%     &\Pr\left\{G_{n, \text{on}}^{(\omega)}(\theta_n) \geq \frac{1}{\alpha}\right\} \\
%     &\geq \Pr\left\{\frac{1}{n}\sum_{i=1}^n \ell(\thetahat_{i-1}\altgiven Z_i) - \ell(\thetahat_n\altgiven Z_i)\leq \frac{\Delta_n}{2}\right\} + \Pr\left\{\rhat_n(\thetahat_n) - R(\thetahat_n) \leq \frac{\Delta_n}{2}\right\} + \\&\qquad\qquad \Pr\left\{R(\thetahat_n) - \rhat_n(\theta_n) \leq \frac{\log\alpha}{n\omega} - \Delta_n\right\} - 2
% \end{align*}
% where we now choose $\Delta_n = (R(\theta_n) - \inf_{\vartheta}R(\vartheta))/2$.
We must show that each term has limit $1$. That the first addend converges in probability to $1$ follows by essentially the same argument as in the previous theorem: We have by \Cref{lem:onlinedomoffline} that
\begin{equation*}
    \Pr(\sigma_n \leq \frac{\Delta_n}{2})  \geq \Pr[\frac{1}{n}\sum_{i=1}^n\omegahat_{i-1} \left\{\ell(\thetahat_{i-1}\altgiven Z_i) - \ell(\thetahat_i\altgiven Z_i)\right\} \leq \frac{\Delta_n}{2} -\overline{\Omega}\cdot \frac{H_n^{(\varepsilon)}}{n}\delta]
\end{equation*}
and by stability, there exists a sequence $\gamma_n = o(n^{-\beta})$ such that for each $n\in\N$, $\ell(\thetahat_{n-1}\altgiven Z_n) - \ell(\thetahat_n \altgiven Z_n) \leq \gamma_n$. Hence, since $\frac{1}{n}\sum \gamma_i = o(n^{-\beta})$ and $\Delta_n \geq c\cdot n^{-\beta}$, we have that
\begin{equation*}
    \lim_{n\rightarrow\infty}\Pr(\sigma_n \leq \frac{\Delta_n}{2}-\overline{\Omega}\cdot\frac{H_n^{(\varepsilon)}}{n}\delta)  \geq \lim_{n\rightarrow\infty}\Pr\left\{\frac{1}{n}\sum_{i=1}^n \gamma_n \leq \frac{\Delta_n}{2}-\overline{\Omega}\cdot \frac{H_n^{(\varepsilon)}}{n}\delta \right\} = 1
\end{equation*}
as required, since $\frac{\Delta_n}{2}-\overline{\Omega} H_n^{(\varepsilon)}n^{-1}\delta = \Omega(n^{-\beta})$.

For the second addend, we again use essentially the same argument as the previous theorem to lower bound $\Pr\left\{\Phi_n(\thetahat_n) < \frac{\Delta_n}{2} \right\}$ by
\begin{equation*}
\Pr\left\{\frac{\overline{\Omega} N}{n}\rhat_N(\thetahat_N) < \frac{\Delta_n}{8}\right\} + \Pr\left\{\overline{\Omega} N^2 \cdot \frac{1}{n}\sum_{i=1}^n \beta_i < \frac{\Delta_n}{8}\right\} + \Pr\left\{\frac{\Delta_n}{8 \inf_\vartheta R(\vartheta)}\rhat_n(\thetahat_n) < \frac{\Delta_n}{4}\right\} - 2.
\end{equation*}
where as before, $N$ is some large enough integer and the $\beta_i$ are some positive constants from the stability hypothesis. Once again, it is easy to see that each of these terms has limit 1.

To show that the third addend has limit $1$, we can again use similar techniques as in the \hyperref[prf:thmpower]{proof of Theorem 3} to lower bound the addend by
\begin{equation*}
    \Pr\left\{R(\thetahat_n) - \inf_{\vartheta}R(\vartheta) \leq \overline{\Omega}^{-1}\Delta_n\right\} + \Pr\left\{\omega R(\theta_n) - \frac{1}{n}\sum_{i=1}^n \omegahat_{i-1}\ell(\theta_n\altgiven Z_i) \leq \frac{\log \alpha}{n} + \Delta_n\right\} - 1.
\end{equation*}
For the former term, note that if $\sup_{\vartheta}|\Phi(\vartheta)| < \varepsilon n^{-\beta}$ for some $\varepsilon > 0$, then
\begin{equation*}
    R(\thetahat_n) \leq \rhat_n(\thetahat_n) +  \varepsilon n^{-\beta} \leq \rhat_n(\theta) +  \varepsilon n^{-\beta} \leq R(\theta) + 2 \varepsilon n^{-\beta}
\end{equation*}
for any $\theta\in\Theta$. Taking the infimum over $\theta$, we arrive at the implication
\begin{equation*}
    \sup_{\vartheta}|\Phi(\vartheta)| < \varepsilon n^{-\beta} \implies R(\thetahat_n) \leq \inf_\vartheta R(\vartheta) + 2\varepsilon n^{-\beta},
\end{equation*}
and so 
\begin{equation}\label[inequality]{ineq:thm3_2}
   \Pr\left\{R(\thetahat_n) \leq \inf_\vartheta R(\vartheta) + 2\varepsilon n^{-\beta}\right\} \geq \Pr\left\{\sup_{\vartheta}|\Phi(\vartheta)| \leq \varepsilon n^{-\beta}\right\} .
\end{equation}
Since the right hand side of (\ref{ineq:thm3_2}) has limit $1$ by hypothesis, using $\varepsilon = \overline{\Omega}^{-1} \cdot c$ yields the result for the former term. For the latter term, we can again use the same previously demonstrated techniques.

\subsubsection{Offline GUe (Fixed Validation Set)}
For convenience in notation, we define $\Phi(\theta) := \rhat_{S_2}(\theta) - R(\theta)$; we also define $\Delta_{n_2} = [R(\theta_{n_2}) - \inf_\vartheta R(\vartheta)]/3$, so that there exists some $c > 0$ such that $\Delta_{n_2} \geq cn^{-\beta}/3$. Then in the same manner as in \Cref{part:thm2offline} of the \hyperref[prf:thmpower]{proof of Theorem 3}, we have that
\begin{align*}
    &\Pr\left\{\widehat{G}_{S, \text{off}}(\theta_n) \geq 1/\alpha\right\}
    \\&\geq \Pr\left\{\Phi(\thetahat_{S_1}) \leq \Delta_{n_2}\right\}
    + \Pr\left\{R(\thetahat_{S_1}) - R(\theta_n) \leq \frac{\log \alpha}{\omegahat_{S_1} n_2} - 2\Delta_{n_2}\right\} + \Pr\left\{-\Phi(\theta_n) \leq \Delta_{n_2}\right\} - 2.
\end{align*}
As usual, we show that each term limit 1. 

For the first term, we have that
% \begin{align*}
%     \lim_{n_2 \rightarrow \infty}\Pr[\rhat_{S_2}(\thetahat_{S_1}) - R(\thetahat_{S_1}) \leq \Delta_{n_2}]
%     &=  \E\qty[\lim_{n_2\rightarrow\infty} \Pr[\rhat_{S_2}(\thetahat_{S_1}) - R(\thetahat_{S_1}) \leq \Delta_{n_2} \mid S_1]]\\
%     &\geq  \E\qty[\lim_{n_2\rightarrow\infty} \Pr[\rhat_{S_2}(\thetahat_{S_1}) - R(\thetahat_{S_1}) \leq \frac{c}{3n_2^\beta} \mid S_1]]\\
%     &= 1
% \end{align*}
% where the first equality is identical to what was seen in the proof of \Cref{thm:power}, the inequality is since $\Delta_{n_2} \geq cn_2^{-\beta}/3$, and the last equality is since $\sup_{\theta}|\rhat_{S_2}(\theta) - R(\theta)| = o_p(n_2^{-\beta})$.
\begin{equation*}
    \lim_{n_2 \rightarrow \infty}\Pr\left\{\rhat_{S_2}(\thetahat_{S_1}) - R(\thetahat_{S_1}) \leq \Delta_{n_2}\right\} 
    \geq \lim_{n_2 \rightarrow \infty}\Pr\left\{\rhat_{S_2}(\thetahat_{S_1}) - R(\thetahat_{S_1}) \leq \frac{c}{3n_2^\beta}\right\} = 1
\end{equation*}
where the convergence is uniform by essentially the same arguments as in in the \hyperref[prf:thmpower]{proof of Theorem 3}, but we now use the fact that $\sup_{\vartheta}|\rhat_{S}(\vartheta) - R(\vartheta)|$ is $o_p(n^{-\beta})$ rather than simply $o_p(1)$.

For the second term, we note that
\begin{align*}
    \Pr\left\{R(\thetahat_{S_1}) - R(\theta) \leq \frac{\log \alpha}{\omegahat_{S_1} n_2} - 2\Delta_{n_2}\right\}
    &= \Pr\left\{R(\thetahat_{S_1}) - \inf_\vartheta R(\vartheta) \leq \frac{\log \alpha}{\omegahat_{S_1} n_2} + \Delta_{n_2}\right\} \\
    &\geq  \Pr\left\{R(\thetahat_{S_1}) - \inf_\vartheta R(\vartheta) \leq \frac{\log \alpha}{\omegahat_{S_1} n_2} + \frac{c}{3n_2^\beta}\right\}.
\end{align*}
First note that there exists $M\in\Z^+$ such that if $n_1 \geq M$ then $\omegahat_{S_1} \geq \omegaunder/2$. Then, we notice that $2\log(\alpha)/(\omegaunder n_2) + c/(3n_2^\beta) > 0$ if and only if $n_2^{1-\beta} > 6\log(1/\alpha)/(c\omegaunder)$; since $\beta\in(0, 1)$ there exists some $N_2\in\Z^+$ such that the right hand side is positive for all $n_2\geq N_2$. Next, since $R(\thetahat_{S_1}) \ip \inf_\vartheta R(\vartheta)$, we know that there exists $N_1\in\Z^+$ such that for any $\varepsilon > 0$,
\begin{equation*}
    \Pr\left\{R(\thetahat_{S_1}) - \inf_\vartheta R(\vartheta) \leq \frac{2\log \alpha}{\omegaunder N_2} + \frac{c}{3N_2^\beta}\right\} > 1-\varepsilon
\end{equation*}
for all $n_1 \geq N_1$. We thus have that for any $\varepsilon > 0$, if $n_1, n_2 \geq \max(N_1, N_2, M)$, then 
\begin{equation*}
    \Pr\left\{R(\thetahat_{S_1}) - R(\theta) \leq \frac{\log \alpha}{\omegahat_{S_1} n_2} - 2\Delta_{n_2}\right\}  > 1-\varepsilon
\end{equation*}
so the second addend has double limit $1$ again.

For the third addend, we simply apply our uniform convergence in probability at rate $n_2^{-\beta}$ since $\Delta_{n_2} \geq cn_2^{-\beta}/3$. 

\subsubsection{Offline GUe (Growing Validation Set)}
We define $\Phi$ as in the previous part, and we also define $\Delta_n := [R(\theta_n) - \inf_\vartheta R(\vartheta)]/3$ so that there exists $c > 0$ such that $\Delta_n > cn^{-\beta}/3$. Then as in the previous parts,
\begin{align*}
    &\Pr\left\{\widehat{G}_S(\theta_n) \geq 1/\alpha\right\}
    \\&\geq \Pr\left\{\Phi(\thetahat_{S_1}) \leq \Delta_n\right\}
    + \Pr\left\{R(\thetahat_{S_1}) - R(\theta_n) \leq \frac{\log \alpha}{\omegahat_{S_1} n_2} - 2\Delta_n\right\} + \Pr\left\{-\Phi(\theta_n) \leq \Delta_n\right\} - 2.
\end{align*}
% \begin{align*}
%     &\Pr\left\{G_S^{(\omega)}(\theta_n) \geq 1/\alpha\right\}
%     \\&\geq \Pr\left\{\rhat_{S_2}(\thetahat_{S_1}) - R(\thetahat_{S_1}) \leq \Delta_n\right\}
%     + \Pr\left\{R(\thetahat_{S_1}) - R(\theta_n) \leq \frac{\log \alpha}{\omega n_2} - 2\Delta_n\right\} \\
%     &\qquad+ \Pr\left\{R(\theta_n)- \rhat_{S_2}(\theta_n) \leq \Delta_n\right\} - 2.
% \end{align*}
and we again show that each term has limit $1$. 

For the first addend, since $n_1 \lesssim n_2$, there exist $k > 0$ and $N_1 \in \Z^+$ such that if $n_1 \geq N_1$, then $n_1 \leq k\cdot n_2$. Furthermore, we have from uniform convergence of $\rhat_{S_2}$ to $R$ at rate $o_p(n_2^{-\beta})$ that for every $\varepsilon > 0$, there exists $N_2$ such that if $n_2 \geq N_2$,
\begin{equation*}
    1 - \Pr\left\{\sup_{\vartheta\in\Theta}|\Phi(\vartheta)| \leq \frac{c}{3((k+1)n_2)^\beta}\right\} < \varepsilon
\end{equation*}
for some $c> 0$. Similarly to \Cref{part:thm2offline} in the \hyperref[prf:thmpower]{proof of Theorem 3}, we then have that for any $\varepsilon > 0$, there exists $N_2$ such that for any $n_1 \geq N_1$, if $n_2\geq N_2$
\begin{equation}\label{eq:intermed}
    1 - \Pr\left\{\Phi(\thetahat_{S_1}) \leq \frac{c}{3((k+1)n_2)^\beta}\right\} < \varepsilon.
\end{equation}
But when $n_1 \geq N_1$, we have that
\begin{equation*}
    \Delta_n \geq \frac{c}{3n^\beta} = \frac{c}{3(n_1 + n_2)^\beta} \geq \frac{c}{3(k+1)n_2)^\beta}
\end{equation*}
and so \cref{eq:intermed} reduces to
\begin{equation*}
    1 - \Pr\left\{\Phi(\thetahat_{S_1}) \leq \Delta_n\right\} < \varepsilon.
\end{equation*}
We hence have that
\begin{equation*}
    \lim_{n_2\rightarrow \infty}\Pr\left\{\Phi(\thetahat_{S_1}) \leq \Delta_n\right\}  = 1
\end{equation*}
and the limit is uniform in $n_1$, as necessary for the double limit to exist and equal $1$. 
% \begin{align*}
%     \lim_{n_2 \rightarrow \infty}\Pr[\rhat_{S_2}(\thetahat_{S_1}) - R(\thetahat_{S_1}) \leq \Delta_n]
%     &=  \E\qty[\lim_{n_2\rightarrow\infty} \Pr[\rhat_{S_2}(\thetahat_{S_1}) - R(\thetahat_{S_1}) \leq \Delta_n \mid S_1]]\\
%     &\geq  \E\qty[\lim_{n_2\rightarrow\infty} \Pr[\rhat_{S_2}(\thetahat_{S_1}) - R(\thetahat_{S_1}) \leq \frac{c}{3(n_1+n_2)^\beta} \mid S_1]]\\
%     &= 1
% \end{align*}
% where the first equality is identical to what was seen in the proof of \Cref{thm:power}, the inequality is since $\Delta_n \geq cn^{-\beta}/3 = c(n_1+n_2)^{-\beta}/3$, and the last equality is since $\sup_{\theta}|\rhat_{S_2}(\theta) - R(\theta)| = o_p(n_2^{-\beta})$.

For the second addend, we note that
\begin{align*}
    \Pr\left\{R(\thetahat_{S_1}) - R(\theta) \leq \frac{\log \alpha}{\omegahat_{S_1} n_2} - 2\Delta_n\right\}
    &= \Pr\left\{R(\thetahat_{S_1}) - \inf_\vartheta R(\vartheta) \leq \frac{\log \alpha}{\omegahat_{S_1} n_2} + \Delta_n\right\} \\
    &\geq \Pr\left\{R(\thetahat_{S_1}) - \inf_\vartheta R(\vartheta) \leq \frac{\log \alpha}{\omegahat_{S_1} n_2} + \frac{c}{3(n_1 + n_2)^\beta}\right\}.
\end{align*}
Similarly to the first addend, there exist $k>0$ and $N_1\in\Z^+$ such that if $n_1 \geq N_1$, $n_1 \leq k \cdot n_2$ and $\omegahat_{S_1} > \omegaunder/2$. So for all $n_1 \geq N_1$, the above is at least
\begin{equation*}
    \Pr\left\{R(\thetahat_{S_1}) - \inf_\vartheta R(\vartheta) \leq \frac{2\log \alpha}{\omegaunder n_2} + \frac{c}{3((k+1)n_2)^\beta}\right\}.
\end{equation*}
Next, we notice that since $\beta\in(0, 1)$ there exists some $N_2\in\Z^+$ such that the right hand side is positive if $n_2\geq N_2$. Then since $R(\thetahat_{S_1}) \ip \inf_\vartheta R(\vartheta)$, we know that there exists $M\in\Z^+$ such that for any $\varepsilon > 0$,
\begin{equation*}
    \Pr\left\{R(\thetahat_{S_1}) - \inf_\vartheta R(\vartheta) \leq \frac{2\log \alpha}{\omegaunder N_2} + \frac{c}{3(k+1)^\beta N_2^\beta}\right\} > 1-\varepsilon
\end{equation*}
for all $n_1 \geq M$. We thus have that for any $\varepsilon > 0$, if $n_1, n_2 \geq \max(N_1, N_2, M)$, then 
\begin{equation*}
    \Pr\left\{R(\thetahat_{S_1}) - R(\theta) \leq \frac{\log \alpha}{\omegahat_{S_1} n_2} - 2\Delta_n\right\}  > 1-\varepsilon
\end{equation*}
so the second addend has double limit $1$ again.

For the third addend, we simply apply our uniform convergence in probability at rate $n^{-\beta}$ since $\Delta_n \geq cn^{-\beta}/3$. 

\section{Parametric Bootstrap for Learning Rate Calibration}\label{sec:paramboot}
\begin{algorithm}[H]
    \caption{Parametric Bootstrap for Learning Rate Calibration}\label{alg:param}
    \begin{algorithmic}
        \Require $\{\D_\theta\}_{\theta\in\Theta}$, a family of parametric distributions
        \Require $(z_1, \ldots, z_n)$, collected data from $\D_\theta$ for some $\theta$
        \Require $\Omega$, a set of candidate learning rates
        \Require $\alpha$, a significance level to calibrate to
        \Require $N$, the number of bootstrap iterations to do

        \State Compute $\widehat{D}$, the best-fitting distribution for $(z_1, \ldots, z_n)$ from $\{\D_\theta\}_{\theta\in\Theta}$
        \State Compute $\thetahat$, the ERM for $(z_{1}, \ldots, z_{n})$
        \State $\text{coverages}(\omega) \gets 0\textbf{ for all } \omega\in\Omega$
        \For{$\omega\in\Omega$}
            \For{$i$ in $1,\ldots, N$}
                \State Draw $S_B = (z_{b(1)}, \ldots, z_{b(n)}) \sim \widehat{D}^n$
                \If{$G^{(\omega)}_{S_B}(\thetahat) < 1/\alpha$}
                    \State $\text{coverages}(\omega) \gets \text{coverages}(\omega) + 1/N$
                \EndIf
            \EndFor
        \EndFor
        \State \Return $\argmin_{\omega\in\Omega} |\text{coverages}(\omega) - (1-\alpha)|$
    \end{algorithmic}
\end{algorithm}

\section{Results for the $L^2$ loss}

\subsection{Analytical Results for the Learning Rate}\label{sec:l2lranlys}

As mentioned in Section 4 of the main manuscript, closed form-learning rates for the $L^2$ loss function can be derived. For example, we have the following proposition for normally distributed data:
\begin{proposition}\label{prop:normallr}
Suppose $X_1, \ldots, X_{2n} \overset{iid}{\sim} N(\theta^*, \sigma^2)$, and define
\begin{equation*}
    b^{(\omega)}_{\alpha, \sigma^2}(z) := \frac{\log(1/\alpha)}{2\omega\sigma^2 z} + \frac{z}{2}.
\end{equation*}
Then the learning rate $\omega$ for the offline GUe-value that obtains exactly $(1-\alpha)$-level coverage for $\theta^*$ under the $L^2$ loss is given by the solution to the equation
\begin{equation}\label{eq:learningrate}
    \int_0^\infty \int_{b^{(\omega)}_{\alpha, \sigma^2}(z_2)}^\infty \frac{\exp(-\frac{z_1^2 + z_2^2}{2})}{2\pi}\dd{z_1}\dd{z_2} + \int_{-\infty}^0 \int_{-\infty}^{b^{(\omega)}_{\alpha, \sigma^2}(z_2)} \frac{\exp(-\frac{z_1^2 + z_2^2}{2})}{2\pi}\dd{z_1}\dd{z_2} = \alpha
\end{equation}
\end{proposition}
\begin{proof}
We can first verify that the strong central condition holds for normally distributed data. For $Z \sim N(\theta^*, \sigma^2)$ and the $L^2$ loss, we have that
\begin{align*}
    \E\left\{\exp(-\omega[\ell(\theta\altgiven Z) - \ell(\theta^*\altgiven Z)])\right\} 
    &=  \E\left\{\exp(-\omega[(\theta - Z)^2 - (\theta^* - Z)^2])\right\}\\
    &= \int_{-\infty}^\infty \exp(2\omega(\theta - \theta^*)z - \omega(\theta^{2} - \theta^{*2})) \cdot \frac{\exp(-\frac{(z-\theta^*)^2}{2\sigma^2})}{\sqrt{2\pi\sigma^2}} \dd{z} \\
    &= \exp((\theta-\theta^*)^2(2\sigma^2\omega - 1)\omega)
\end{align*}
which is upper bounded by $1$ for all $\theta\in\R$ if $0 < \omega \leq 1/(2\sigma^2)$.

Towards proving the proposition, let $\overline{X}$ denote the sample mean of $X_1, \ldots, X_n$ and $\thetahat$ denote the sample mean of $X_{n+1}, \ldots, X_{2n}$. Then by expanding the definition of the GUe-value and using the law of total probability, we have that
\begin{align*}
    &\Pr\left\{G_S(\theta^*) \geq 1/\alpha\right\} \\
    &= \Pr\left\{\overline{X}(\thetahat - \theta^*) - \frac{(\thetahat - \theta^*)(\thetahat + \theta^*)}{2} \geq \frac{\log(1/\alpha)}{2n\omega}\right\} \\
    &= \frac{\Pr\left\{\overline{X} \geq \frac{\log(1/\alpha)}{2n\omega(\thetahat - \theta^*)} + \frac{\thetahat + \theta^*}{2} \mid \thetahat > \theta^*\right\}}{2} + \frac{\Pr\left\{\overline{X} \leq \frac{\log(1/\alpha)}{2n\omega(\thetahat - \theta^*)} + \frac{\thetahat + \theta^*}{2} \mid \thetahat \leq \theta^*\right\}}{2} \\
    &= \frac{1}{2}\Pr\left\{Z_1 \geq \frac{\log(1/\alpha)}{2\omega\sigma^2 Z_2} + \frac{Z_2}{2} \mid Z_2 > 0\right\} + \frac{1}{2}\Pr\left\{Z_1 \leq \frac{\log(1/\alpha)}{2\omega\sigma^2 Z_2} + \frac{Z_2}{2} \mid Z_2 \leq 0\right\} \\
    &= \frac{1}{2}\Pr\left\{Z_1 \geq b^{(\omega)}_{\alpha, \sigma^2}(Z_2) \mid Z_2 > 0\right\} + \frac{1}{2}\Pr\left\{Z_1 \leq b^{(\omega)}_{\alpha, \sigma^2}(z_2) \mid Z_2 \leq 0\right\}
\end{align*}
where $Z_1 = \sqrt{n}(\overline{X}-\theta^*)/\sigma$ and $Z_2 = \sqrt{n}(\thetahat -  \theta^*)/\sigma$ are independent standard normal random variables, whence the result follows by substituting integrals of standard normal densities for the probability statements.
\end{proof}
We note that \cref{eq:learningrate} can be solved numerically for $\omega$ quite quickly, so it is a convenient choice for learning rate whenever \Cref{prop:normallr} is applicable. The proof of the proposition illustrates that in order to use \cref{eq:learningrate} for non-normal random variables, we need the sample size to be large enough for sample means to be reasonably approximated as normal and for the sample variance to act as a good estimator for $\sigma^2$. If safety at smaller sample sizes is a concern, however, one could simply solve for $\omega$ from \cref{eq:learningrate} and divide the learning rate by two (for example) to be confident that the learning rate is small enough to be a safe choice.

To ensure safety for non-normal data, we may use the following proposition:
\begin{proposition}
Let $X_1, \ldots, X_{2n}$ be i.i.d. from any distribution with mean $\theta^*$, variance $\sigma^2$, and third absolute moment $\rho$. Let $c_B \approx 0.4748$ be the Berry-Esseen constant and $\Phi$ denote the standard normal CDF, and define
\begin{align*}
    l^{(\omega)}_\alpha(z) &:= \max\qty(\Phi(b^{(\omega)}_{\alpha, \sigma^2}(z)) - \frac{c_B\rho}{\sigma^3\sqrt{n}}, 0) \\
    u^{(\omega)}_\alpha(z) &:= \min\qty(\Phi(b^{(\omega)}_{\alpha, \sigma^2}(z)) + \frac{c_B\rho}{\sigma^3\sqrt{n}}, 1).
\end{align*}
Furthermore, let $l^{(\omega)}_\alpha(z)$ be maximized on $[0, \infty)$ at $z = \beta$ and $u^{(\omega)}_\alpha(z)$ be maximized on $(-\infty, 0]$ at $z=\gamma$. 
Then any $\omega$ that satisfies
\begin{align*}
    \alpha &\geq 1 - \max\qty(\frac{1}{2} - \frac{c_B\rho}{\sigma^3\sqrt{n}}, 0) - \max\qty(1 - \frac{c_B\rho}{\sigma^3\sqrt{n}}, 0) \\
    &\qquad+ \qty[\max(1 - \frac{c_B\rho}{\sigma^3\sqrt{n}}, 0) + \min(\frac{c_B\rho}{\sigma^3\sqrt{n}}, 1)]\cdot \min\qty(\frac{1}{2} + \frac{c_B\rho}{\sigma^3\sqrt{n}}, 1) \\
    &\qquad+ \int_0^\beta l^{(\omega)}_\alpha(z) \dd{l^{(\omega)}_\alpha(z)} + \int_\beta^\infty u^{(\omega)}_\alpha(z) \dd{l^{(\omega)}_\alpha(z)} \\
    &\qquad + \int_{-\infty}^\gamma u^{(\omega)}_\alpha(z) \dd{u^{(\omega)}_\alpha(z)} + \int_\gamma^0 l^{(\omega)}_\alpha(z)\dd{u^{(\omega)}_\alpha(z)} 
\end{align*}
obtains at least $(1-\alpha)$-level coverage for $\theta^*$ under the $L^2$ loss using the offline GUe-value. 
\end{proposition}
\begin{proof}
We can follow the proof of \Cref{prop:normallr} up until the point where we assume $Z_1$ and $Z_2$ are standard normal; supposing they instead have CDF $F$, we arrive at
\begin{equation*}
    \Pr\left\{G_S(\theta^*) \geq 1/\alpha\right\} = \int_0^\infty 1 - F(b^{(\omega)}_{\alpha, \sigma^2}(z_2)) \dd{F(z_2)} + \int_{-\infty}^0 F(b^{(\omega)}_{\alpha, \sigma^2}(z_2)) \dd{F(z_2)}.
\end{equation*}
Although we do not know $F$, we do know by the Berry-Esseen inequality that
\begin{equation*}
    \sup_{x\in\R}|F(x) - \Phi(x)| \leq \frac{c_B\rho}{\sigma^3\sqrt{n}}
\end{equation*}
where $\Phi$ denotes the standard normal CDF. For convenience, we suppress all unnecessary parameters so that we denote $b(z) = b^{(\omega)}_{\alpha,\sigma^2}(z)$, $l(z) = l^{(\omega)}_\alpha(z)$, and $u(z) = u^{(\omega)}_\alpha(z)$. Then we can obtain a safe learning rate by repeatedly integrating by parts and applying the Berry-Esseen bounds to upper bound this expression:
\begin{align*}
    &\int_0^\infty 1 - F(b(z)) \dd{F(z)} + \int_{-\infty}^0 F(b(z)) \dd{F(z)} \\
    &=(1 - F(0)) - \int_0^\infty F(b(z)) \dd{F(z)} + \int_{-\infty}^0 F(b(z)) \dd{F(z)} \\
    &\leq 1 - F(0) - \int_0^\infty l(z)\dd{F(z)} + \int_{-\infty}^0 u(z) \dd{F(z)} \\
    &= 1 - F(0) - \qty[l(\infty) - l(0)F(0) - \int_0^\infty F(z) \dd{l(z)}] + \qty[u(0)F(0) + \int_{-\infty}^0 F(z) \dd{u(z)}] \\
    &= 1 - F(0)  - l(\infty) + l(0)F(0) + u(0)F(0) \\
    &\qquad\qquad+ \int_0^\beta F(z) \dd{l(z)} + \int_\beta^\infty F(z) \dd{l(z)} + \int_{-\infty}^\gamma F(z) \dd{u(z)} + \int_\gamma^0 F(z)\dd{u(z)} \\
    &\leq 1 - F(0) - l(\infty) + [l(0) + u(0)] \cdot F(0) \\
    &\qquad\qquad+ \int_0^\beta l(z) \dd{l(z)} + \int_\beta^\infty u(z) \dd{l(z)} + \int_{-\infty}^\gamma u(z) \dd{u(z)} + \int_\gamma^0 l(z)\dd{u(z)} \\
    &\leq 1 - \max\qty(\frac{1}{2} - \frac{c_B\rho}{\sigma^3\sqrt{n}}, 0) - \max\qty(1 - \frac{c_B\rho}{\sigma^3\sqrt{n}}, 0) \\
    &\qquad\qquad+ \qty[\max\qty(1 - \frac{c_B\rho}{\sigma^3\sqrt{n}}, 0) + \min\qty(\frac{c_B\rho}{\sigma^3\sqrt{n}}, 1)]\cdot \min\qty(\frac{1}{2} + \frac{c_B\rho}{\sigma^3\sqrt{n}}, 1) \\
    &\qquad\qquad+ \int_0^\beta l(z) \dd{l(z)} + \int_\beta^\infty u(z) \dd{l(z)} + \int_{-\infty}^\gamma u(z) \dd{u(z)} + \int_\gamma^0 l(z)\dd{u(z)}
\end{align*}
as desired.
\end{proof}
This no longer depends on the distribution of the data other than the second and third moments. Thus, so long as the ratio of these moments can be well-approximated, this proposition yields a provably safe choice of learning rate. We note that it is atypical for there to exist an $\omega$ that obtains exactly $(1-\alpha)$-level coverage from the above proposition---in general, all learning rates satisfying the proposition are more conservative than necessary.

\subsection{Power Results}\label{sec:l2power}
The power results in the main manuscript are fairly weak due to the use of a general loss function. However, when substituting a specific loss function such as the $L^2$ loss, stronger results can be attained. Indeed, GUe confidence sets for the mean are (nearly) asymptotically efficient in the sense that their radius shrinks only slightly more slowly than $n^{-1/2}$.

\begin{proposition}
Let $(\theta_n)_{n\in\mathbb{N}}$ be a sequence in $\Theta$ such that $\theta_n - \theta^* \gtrsim (n/\log n)^{-1/2}$, where the hidden constant in the right-hand side is sufficiently large. If the data-generating distribution has a finite second moment, then under the $L^2$ loss,
\begin{equation*}
    \lim_{n\rightarrow\infty}\Pr\{G_{n,\text{on}}(\theta_n) \geq \alpha^{-1}\} = 1.
\end{equation*}
\end{proposition}
\begin{proof}
With squared error loss, our GUe-process is 
\[ G_{n,\text{on}}(\theta_n) = \exp\Bigl[ \omega \Bigl\{ \underbrace{\sum_{i=1}^n (Z_i - \theta_n)^2 - \sum_{i=1}^n (Z_i - \thetahat_{i-1})^2}_{\text{regret}} \Bigr\} \Bigr], \]
where $\theta_n = \theta^* + a_n n^{-1/2}$ for some deterministic sequence $(a_n)_{n\in\mathbb{N}}$ satisfying $a_n \gtrsim (\log n)^{1/2}$, and $\thetahat_{i-1}$ is the running sample mean, i.e., $\thetahat_k = \frac{1}{k} \sum_{i=1}^k Z_i$. % The goal is to show that, even if $\theta=\theta_n$ is collapsing to $\theta^*$ at rate near $n^{-1/2}$, the GUe-process still blows up to infinity with probability converging to 1.  This implies, for example, that our GUe-process confidence sets are of roughly the optimal size as a function of the sample size, $O(n^{-1/2})$.  Even though we're struggling to prove this as a consequence of our general results, it's pretty easy to get this with a direct analysis. 
We show that the ``regret" term approaches $\infty$ with probability converging to 1. % even if $\theta=\theta_n = \theta^* + a_n n^{-1/2}$, for some deterministic sequence $(a_n)$ with $|a_n| \to \infty$; what I have in mind is taking $a_n \propto \log n$, etc.  
Algebraic manipulations immediately yield an alternative expression for regret:
\[ \underbrace{n(\theta^* - \theta_n)^2 + 2n(\theta^* - \theta_n)(\thetahat_n - \theta^*)\vphantom{\sum_{i=1}^n}}_{\text{(A)}} - \underbrace{\sum_{i=1}^n (\thetahat_{i-1} - \theta^*)^2}_{\text{(B)}} - \underbrace{2\sum_{i=1}^n (\theta^* - \thetahat_{i-1})(Z_i - \theta^*)}_{\text{(C)}}. \]
Note that the first term, (A), can be rewritten as
    $a_n^2 - 2a_nn^{1/2}(\thetahat_n-\theta^*),$
which is $a_n^2 - a_n O_p(1)$ by the central limit theorem; it follows that this first term converges in probability to infinity.

For the second term, (B), note that
\[ \E(\text{B}) = \sum_{i=1}^n \E(\thetahat_{i-1} - \theta^*)^2 = (\thetahat_0 - \theta^*)^2 + \operatorname{Var}(Z_1) \sum_{i=2}^n (i-1)^{-1} = O(\log n). \]
Since (B) is non-negative, it follows immediately from Markov's inequality that (B) is $O_P(\log n)$. 
Finally, (C) has expected value $0$ and 
\[ \E(\text{C}^2) = \sum_{i=1}^n \E(\theta^* - \thetahat_{i-1})^2(Z_i - \theta^*)^2 = (\thetahat_0 - \theta^*)\operatorname{Var}[Z_1] + \sum_{i=2}^n \frac{\operatorname{Var}[Z_i]^2}{i-1} = O(\log n). \]
It is then a consequence of Chebyshev's inequality that $|\text{C}|$ is $O_P(\log n)$.

Putting everything together, the regret is lower-bounded by  $O_p(a_n^2 - a_n - \log n)$, which diverges to $\infty$ provided that $a_n$ is eventually a sufficiently large multiple of $(\log n)^{1/2}$, as desired.
\end{proof}

\section{Additional simulation results}\label{apdx:extrasims}
In this section, we provide three more simulation studies to demonstrate the superior performance of the GUe-based methods compared to existing methods.

\begin{example}\label{ex:heavytail}\normalfont{
We mentioned previously that the strong central condition is a sufficient but perhaps not necessary condition for validity of the proposed GUe tests and confidence sets.  This suggests that there are cases in which the strong central condition fails but GUe still offers valid inference.  
To illustrate this, \Cref{fig:heavytail} shows the GUe confidence sets coverage probabilities on data from the $t$-distribution with $3$ degrees of freedom---a distribution that does not satisfy the strong central condition. For comparison, we also show the coverage of the nonparametric bootstrap, as well as for the Catoni-style confidence interval given in Theorem 9 of \citet{wang2023}. We note that because of the small sample size, the bootstrap greatly undercovers the true mean. Furthermore, although the Catoni-style confidence interval was the most efficient of the anytime-valid confidence sets for heavy-tailed means studied in \citet{wang2023}, it still has an observed coverage near 100\% at all relevant significance levels, agreeing with the results of Figure 4 of \citet{wang2023}. In comparison, both GUe confidence sets maintain very reasonable empirical coverage rates across the entire range.
\begin{figure}[t]
    \centering
    \includegraphics[width = 0.5\textwidth]{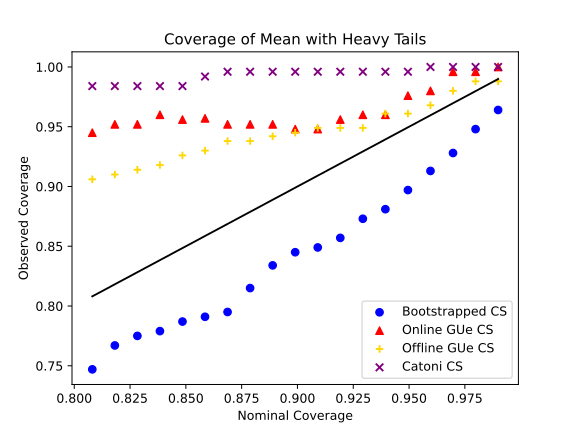}
    \caption{Coverage probabilities for the mean of the $t$-distribution with $3$ degrees of freedom, with $n = 10$. For the Catoni intervals, we follow \citet{wang2023} and use the optimal $\lambda_t$ given by their equation (33) when it exists and, otherwise, use $\lambda_t = \min(t^{-1/2}, 0.1)$.}
    \label{fig:heavytail}
\end{figure}
}\end{example}

\begin{example}\normalfont{
One set of problems that proves challenging for many classical inferential techniques is that of providing uncertainty quantification for parameters that lie at the boundary of the parameter space. More formally, suppose that it is known that the risk minimizing $\theta^*$ lies in a restricted subset $U$ of the parameter space $\Theta$; then we define our loss function to be such that $\ell(\theta\altgiven z) = \infty$ for all $\theta\not\in U$. It is well known that in this sort of problem (even if performing likelihood-based inference rather than loss-function based inference), many classical techniques, such as the bootstrap, fail to provide valid confidence sets when $\theta^*$ lies on the boundary of $U$.

Of course, the GUe-value does not face any difficulties with inference in a parameter-at-boundary problem. To illustrate, we perform quantile estimation where the relevant quantile is at the boundary. In particular, we estimate the $84$th percentile of $\operatorname{N}(-3, 3^2)$ data where we restrict inference on the interval $[0, \infty)$, which is induced by the loss function
\begin{equation*}
    \ell(\theta\altgiven z) = \begin{cases}(z-\theta) \cdot \{q - \mathbbm{1}(z < \theta)\} & \text{if }\theta \geq 0 \\
    \infty & \text{otherwise}
    \end{cases}
\end{equation*}
where $q = 0.84$. \Cref{fig:boundary} illustrates that, as expected, the nonparametric bootstrap fails to attain the nominal level of coverage, as the $0.84$ quantile of the data is exactly $0$, whereas the online and offline GUe confidence sets (selecting the learning rate through the nonparametric bootstrap) do attain above-nominal coverage levels.
\begin{figure}[tbp]
    \centering
    \includegraphics[width=0.5\linewidth]{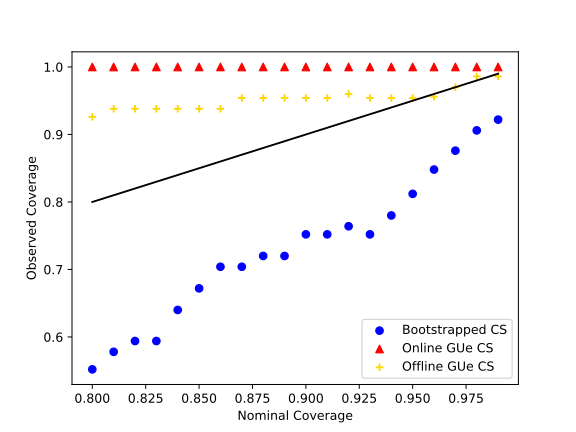}
    \caption{Coverage of the $0.84$ quantile of $\operatorname{N}(-3, 3^2)$ data when the estimate for the quantile is restricted to the interval $[0, \infty)$.}
    \label{fig:boundary}
\end{figure}}
\end{example}

\begin{example}\normalfont{
Suppose that $\vb*{Z}$ is a matrix of random variables such that the distribution of $\vb*{Z}$ is invariant to permutations of its rows and columns. It is shown in \citet[Section 4.3]{mccullagh2000} that the use of nonparametric bootstrap to estimate $\E(\vb*{Z}_{11})$ results in confidence sets that are smaller than necessary to obtain the nominal level of coverage.

To illustrate the simplest possible example of this, consider the two-way ANOVA model with random effects with only one observation per cell:
$
    Z_{ij} = \mu + \alpha_i + \beta_j + \varepsilon_{ij}
$
where the variance components are independent and follow centered exponential distributions---i.e., $\alpha_i + \sigma_\alpha \iid \operatorname{Exp}(\sigma_\alpha)$, $\beta_j + \sigma_\beta \iid \operatorname{Exp}(\sigma_\beta)$, and $\varepsilon_{ij} + \sigma\iid \operatorname{Exp}(\sigma)$. \Cref{fig:anova} demonstrates that once again, the GUe confidence sets using nonparametric bootstrap for learning rate selection obtain far above the nominal level of coverage for $\mu$ whereas the nonparametric bootstrap by itself fails to do so. Additionally, \Cref{fig:anova} provides visualizations of the confidence intervals produced by the studied methods for 10 random datasets; it is clear that the offline GUe confidence intervals are often of comparable width to bootstrapping despite having much better coverage and the desirable robustness properties of e-values. 
\begin{figure}[tbp]
    \centering
    \includegraphics[width=0.48\linewidth]{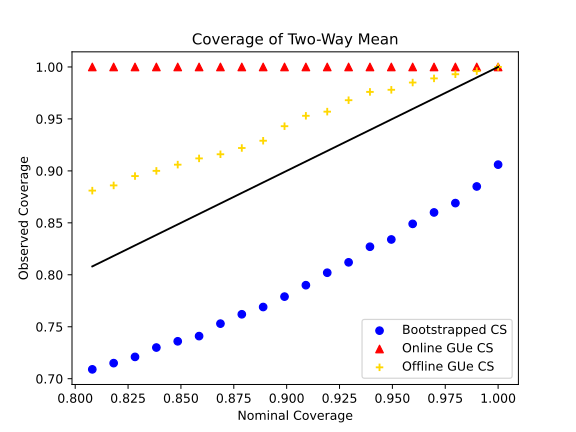}\includegraphics[width=0.48\linewidth]{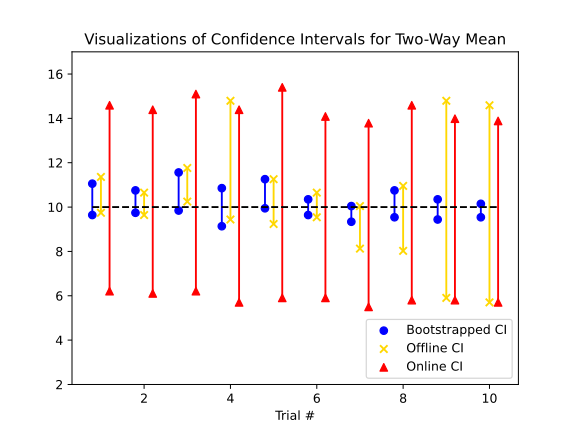}
    \caption{Left: Coverage probabilities for the mean of a two-way ANOVA with random effects and one observation per cell. Right: visualizations of 90\% confidence intervals. Here $n = 10$, $\mu = 10$, $\var(\alpha_i) = 0.1$, $\var(\beta_j) = 0.05$, and $\var(\varepsilon_{ij}) = 1$}
    \label{fig:anova}
\end{figure}

}\end{example}

\end{document}